\documentclass{article} 
\usepackage{nips15submit_e,times}
\usepackage{hyperref,amsmath,amssymb,enumerate, algorithm, algorithmic, subfig, xspace, wrapfig, graphicx, caption}

\usepackage{amsthm}
\usepackage{url}
\usepackage{etoolbox}
\patchcmd{\thebibliography}{\section*{\refname}}{}{}{}

\renewcommand{\algorithmicrequire}{\textbf{Input:}}
\renewcommand{\algorithmicensure}{\textbf{Output:}}

\newcommand{\ice}{{\sc backShift }}
\newcommand{\ling}{{\sc Ling }}
\newcommand{\lingam}{{\sc Lingam }}
\newcommand{\ffdiag}{{\sc FFDiag }}
\newcommand{\ices}{{\sc backShift}}
\newcommand{\lings}{{\sc Ling}}

\newcommand{\iMult}{{m_I}}

\newcommand{\iid}{i.i.d.\ }

\newcommand{\x}{{\mathbf x}}
\newcommand{\X}{{\mathbf X}}
\newcommand{\ci}{{\mathbf c}}
\newcommand{\B}{{\mathbf B}}

\newcommand{\D}{{\mathbf D}}
\newcommand{\J}{{\mathcal J}}
\newcommand{\N}{{\mathcal N}}
\newcommand{\Se}{{\mathbf {\Sigma_e}}}

\newcommand{\Scj}{{\mathbf {\Sigma}_{\mb c,j}}}
\newcommand{\Scjm}{{\mathbf {\Sigma}_{\mb c,-j}}}
\newcommand{\CP}{{\mathit{CP}}}

\newcommand{\DScj}{{\mathbf {\Delta \Sigma}_{\mb c,j}}}

\newcommand{\DScjp}{{\mathbf {\Delta \Sigma}_{\mb c,j'}}}

\newcommand{\DSxj}{{\mathbf {\Delta \Sigma}_{\mb x,j}}}
\newcommand{\sic}{\boldsymbol \eta}

\newcommand{\Ge}{{\mathbf {G_e}}}

\newcommand{\Gcj}{{\mathbf {G}_{\mb c,j}}}
\newcommand{\Si}{{\mathbf {\Sigma}}}
\newcommand{\DSi}{{\mathbf {\Delta \Sigma}}}
\newcommand{\G}{{\mathbf {G}}}

\newcommand{\mb}{\mathbf}

\newcommand{\cov}{\mathrm{Cov}}

\newtheorem{theorem}{Theorem}
\newtheorem{cor}[theorem]{Corollary}

\usepackage[framemethod=TikZ]{mdframed}
\usetikzlibrary{arrows}
\newdimen\arrowsize
\pgfarrowsdeclare{arcsq}{arcsq}
{
  \arrowsize=0.2pt
  \advance\arrowsize by .5\pgflinewidth
  \pgfarrowsleftextend{-4\arrowsize-.5\pgflinewidth}
  \pgfarrowsrightextend{.5\pgflinewidth}
}
{
  \arrowsize=1.5pt
  \advance\arrowsize by .5\pgflinewidth
  \pgfsetdash{}{0pt} 
  \pgfsetroundjoin   
  \pgfsetroundcap    
  \pgfpathmoveto{\pgfpoint{0\arrowsize}{0\arrowsize}}
  \pgfpatharc{-90}{-140}{4\arrowsize}
  \pgfusepathqstroke
  \pgfpathmoveto{\pgfpointorigin}
  \pgfpatharc{90}{140}{4\arrowsize}
  \pgfusepathqstroke
}

\title{\ices: Learning causal cyclic graphs from unknown shift interventions}

\author{
Dominik Rothenh\"ausler$^{*}$\\
Seminar f\"ur Statistik\\
ETH Z\"urich, Switzerland\\
\texttt{rothenhaeusler@stat.math.ethz.ch} \\ 
\And
Christina Heinze\thanks{Authors contributed equally.}\\
Seminar f\"ur Statistik\\
ETH Z\"urich, Switzerland\\
\texttt{heinze@stat.math.ethz.ch} \\
\And
Jonas Peters\\
Max Planck Institute for Intelligent Systems\\
T\"ubingen, Germany \\
\texttt{jonas.peters@tuebingen.mpg.de} \\
\And
Nicolai Meinshausen\\
Seminar f\"ur Statistik\\
ETH Z\"urich, Switzerland\\
\texttt{meinshausen@stat.math.ethz.ch} 
}

\nipsfinalcopy 

\begin{document}

\maketitle

\begin{abstract}
We propose a simple method to learn linear causal cyclic models
 in the presence of latent
variables. The method relies on equilibrium data of the model recorded under a specific kind of interventions (``shift interventions''). The  location  and strength of these interventions do not have to be
known and can be estimated from the data.  Our method, called \ices, only uses second
moments of the data and performs simple joint matrix
diagonalization, applied to differences between covariance matrices. 
We give a sufficient and necessary condition
for identifiability of the system, which is fulfilled almost surely
under some quite general assumptions 
 if and only if there are at least  three distinct experimental
settings, one of which can be pure observational data.
We demonstrate the performance on some simulated data and applications
in flow cytometry and financial time series. 
\end{abstract}

\vspace{-0.1cm}
\section{Introduction}
\vspace{-0.1cm}

Discovering causal effects is a fundamentally important yet very challenging task  in various disciplines, from public health research and sociological studies, economics to many applications in the life sciences. There has been much progress on learning acyclic graphs in the context of structural equation models \cite{Bollen1989}, including  methods that learn from observational data alone under a faithfulness assumption  \cite{Spirtes2000,Chickering2002,Maathuis2009,Hauser2012}, exploiting non-Gaussianity of the data \cite{Hoyer2008,Shimizu2011} or non-linearities \cite{Mooij2011}.
Feedbacks are prevalent in most applications, and we are interested in the setting of  \cite{Hyttinen2012}, where we observe the equilibrium data of a model that is characterized by a set of linear relations
\begin{equation}\label{eq:model}  \x = \B \x + \mb e, \end{equation}
where $\x\in \mathbb{R}^p$ is a random vector and 
$\B \in \mathbb{R}^{p\times p}$ is the connectivity matrix with zeros on the diagonal (no self-loops). Allowing for self-loops would lead to an identifiability problem, independent of the method. See Section~\ref{sec:supp_algo} in the Appendix for more details on this setting.
The graph corresponding to $\B$ has $p$ nodes and an edge from node $j$ to node $i$ 
if and only if $\B_{i,j} \neq 0$. 
The error terms $\mb e$ are $p$-dimensional random variables with mean 0 and positive semi-definite covariance matrix $\Se=E(\mb e \mb e^T)$. We do not assume that $\Se$ is a diagonal matrix which allows the existence of latent variables.

The solutions to \eqref{eq:model} can be thought of as the
deterministic equilibrium solutions (conditional on the noise term)
of a  dynamic model governed by first-order difference equations with
matrix $\B$ in the sense of \cite{lauritzen2002chain}. For well-defined equilibrium solutions of~\eqref{eq:model}, we need that $\mb I-\mb B$ is invertible. Usually we also want 
\eqref{eq:model} to converge to an equilibrium when iterating as $\x^{\mathit (new)}
\leftarrow \B \x^{\mathit (old)} + \mb
e$ or in other words $\lim_{m\to\infty} \B^m \equiv \mb 0$. This condition is
equivalent to the spectral radius of $\B$ being
strictly smaller than one \cite{lacerda2012discovering}.
We will make an assumption on cyclic graphs that restricts
the strength of the feedback.  
 Specifically,  let a
cycle of length $\eta$ be given by $(m_{1}, \ldots,
m_{\eta+1}=m_1) \in \{1,\ldots,p\}^{1+\eta}$ and $m_{k} \neq m_{\ell}$ for $1 \le k < \ell \le \eta$.
We define the cycle-product
$\CP(\B)$ of a matrix $\B$ to be the maximum over cycles of all
lengths $1<\eta\le p$ of the path-products 
\begin{equation}\label{CP}
\CP(\B):= \max_{ \substack{(m_1,\ldots,m_\eta, m_{\eta+1}) \text{ cycle}\\[1pt]1<\eta\le p  }} \; \,\prod_{1\le k \le \eta} \left| \mb B_{m_{k+1},m_{k}} \right|.
\end{equation}

The cycle-product $\CP(\B)$ is clearly zero for acyclic graphs. We will assume the
cycle-product to be strictly smaller than one for identifiability
results, see Assumption~(A) below. 
The most interesting graphs are those for which $\CP(\B)<1$  and for
which the spectral radius of  $\B$  is
strictly smaller than one. Note that these two conditions are identical
as long as the cycles in the graph do not intersect, i.e., there is no node that is part of two cycles (for example if
there is  at most  one cycle in the graph). If cycles do intersect, we
can have models for which either (i) $\CP(\B)<1$ but the spectral radius
is larger than one or (ii) $\CP(\B)>1$ but the spectral radius is
strictly smaller than one. Models in situation (ii) are not stable  in the sense that the iterations will not
converge under interventions. 
We can for example block all but one cycle.
If this one single unblocked cycle has a cycle-product larger than 1 (and there is such a cycle in
the graph if $\CP(\B)>1$), then the solutions of the iteration are not
stable\footnote{The blocking of all but one cycle can be achieved by do-interventions on appropriate variables under the following condition: for every pair of cycles in the graph, the variables in one cycle cannot be a subset of the variables in the other cycle.  Otherwise the blocking could be achieved by deletion of appropriate edges.}.
Models in
situation (i) are
not stable either, even in the absence of interventions. We can still in theory obtain the now instable equilibrium
solutions to~\eqref{eq:model} as $(\mb I - \B)^{-1} \mb e$ and the theory below applies to these instable equilibrium
solutions. However, such instable equilibrium solutions are arguably
of little practical interest. In summary: all interesting feedback
models that are stable under interventions satisfy both
$\CP(\B)<1$ and have a spectral radius strictly smaller than
one. We will just  assume $\CP(\B)<1$ for the following results.

It is impossible to learn the structure $\B$ of this model from observational data alone without making further assumptions.  The \lingam approach has been extended in \cite{lacerda2012discovering}  to cyclic models, exploiting a possible non-Gaussianity of the data.
Using both experimental and interventional data, \cite{scheines2010combining,Hyttinen2012} could show identifiability of the connectivity  matrix $\B$ under a learning mechanism that relies on data under so-called ``surgical'' or ``perfect'' interventions
. In their framework, a variable becomes independent of all its parents if it is being intervened on and all incoming contributions are thus effectively removed under the intervention (also called do-interventions in the classical sense of \cite{Pearl2009}). The learning mechanism makes then use of the knowledge where these ``surgical'' interventions occurred. 
\cite{Eberhardt2010} 
also allow for ``changing'' the incoming arrows for
variables that are intervened on; but again, \cite{Eberhardt2010} requires the location of the interventions 
while 
we do not assume such knowledge. \cite{peters2015causal} consider a target variable and allow for arbitrary interventions on all other nodes. They neither permit hidden variables nor cycles.

Here, we are interested in a setting where we have either no or just very limited knowledge about the exact location and strength of the interventions, as is often the case for data observed under different environments (see the example on financial time series further below) or for biological data \cite{Jackson2003, kulkarni2006evidence}. These interventions have been called ``fat-hand'' or ``uncertain'' interventions \cite{eaton2007exact}. While \cite{eaton2007exact} assume acyclicity and model the structure explicitly in a Bayesian setting,  we assume that the data in environment $j$ are equilibrium observations of the model
\begin{equation}\label{eq:modelinterv}  
\x_j = \B \x_j + \mb c_j + \mb e_j, 
\end{equation}
where the random intervention shift $\ci_j$ has a mean and covariance $\mb \Sigma_{\mb c,j}$. 
The \emph{location} of these interventions (or simply the \emph{intervened variables}) are those 
components of $\mb c_j$ that are not zero with probability one.
Given these locations, the interventions simply shift the variables by a value determined by $\mb c_j$;
they are therefore not ``surgical'' 
but can be seen as a special case of what is called an ``imperfect'', ``parametric'' \cite{Eberhardt2007} or ``dependent'' intervention \cite{Korb2004} or
``mechanism change'' \cite{Tian2001}. The 
matrix $\B$ and the 
error distribution of $\mb e_j$ are assumed
to be identical in all environments. 
In contrast to the covariance matrix for the noise term $\mb e_j$, we \emph{do} assume that 
$\mb \Sigma_{\mb c,j}$
is a diagonal matrix, which is equivalent to demanding that interventions at different variables are uncorrelated. This is a key assumption necessary to identify the model using experimental data. 
Furthermore, we will discuss in Section~\ref{sec:sachs} how a
violation of the model assumption~\eqref{eq:modelinterv} can be
detected and used to estimate the location of the interventions.

In Section~\ref{sec:method} we show how to leverage
observations under different environments
with different interventional distributions to learn the structure of
the connectivity matrix $\B$ in model \eqref{eq:modelinterv}. The
method rests on a simple joint matrix diagonalization. We will prove
necessary and sufficient conditions for identifiability 
in Section~\ref{sec:identifiability}. Numerical results for simulated
data and applications in flow cytometry and financial data are shown in Section~\ref{sec:numerical}.

\vspace{-0.1cm}
\section{Method}\label{sec:method}
\vspace{-0.1cm}
\subsection{Grouping of data}
\vspace{-0.1cm}
Let $\J$ be the set of experimental conditions under which we observe
equilibrium data from model~\eqref{eq:modelinterv}. These different
experimental conditions can arise in two ways: (a) a controlled
experiment was conducted where the external input or the external imperfect
interventions have been deliberately changed from one member of $\J$
to the next. An example are the flow cytometry data
\cite{sachs2005causal} discussed in Section~\ref{sec:sachs}. (b) The data are recorded over
time. It is assumed that the external input is changing over time
but not in an explicitly controlled way. The data are grouped into 
consecutive blocks $j \in \J$ of observations, see Section~\ref{sec:finance} for an example. 

\vspace{-0.1cm}
\subsection{Notation}
\vspace{-0.1cm}
Assume we have $n_j$ observations in each setting $j\in \J$.
Let $\X_j$ be the $(n_j\times p)$-matrix of observations from
model~\eqref{eq:modelinterv}. 
For general random variables $\mb a_j \in
\mathbb{R}^p$
, the population covariance matrix in setting $j\in \J$ is called 
$\Si_{\mb a,j} = \cov( \mb a_j )$, 
where the covariance is under the setting $j\in \J$.  Furthermore, the covariance on all settings except
setting $j\in\J$ is defined as an average over all environments except
for the $j$-th environment, 
$ (|\J|-1) \Scjm :=\sum_{j'\in \J\setminus\{j\}} \Si_{\mb c, j'} .$ The
population Gram matrix is defined as 
$\G_{\mb a,j} = E( \mb a_j {\mb a_j}^T)$.
Let the $(p\times p)$-dimensional $\hat{\Si}_{\mb a,j}$ be the empirical
covariance matrix of the observations $\mb A_j \in \mathbb{R}^{n_j\times
  p}$ of variable $\mb a_j$ in setting $j\in \J$. More precisely, let
$\tilde{\mb A}_j$ be the column-wise mean-centered version of $\mb A_j$.
Then  $\hat{\Si}_{\mb a,j}:=  (n_j-1)^{-1}\tilde{\mb A}_j^T \tilde{\mb
  A}_j$. 
The
empirical Gram matrix is denoted by $\hat{\G}_{\mb a,j}:= n_j^{-1}{\mb A}_j^T
{\mb A}_j$.
\vspace{-0.1cm}
\subsection{Assumptions}
\vspace{-0.1cm}
The main assumptions have been stated already but we give a summary
below. 
\begin{enumerate}[(A)]
\item The data are observations of the equilibrium observations of
  model~\eqref{eq:modelinterv}. The matrix $\mb I-\B$ is invertible
  and the solutions to~\eqref{eq:modelinterv} are thus well defined. The cycle-product~\eqref{CP} $\CP(\B)$ is
  strictly smaller than one. The diagonal entries of $\B$ are zero.
\item The distribution of the noise $\mb e_j$ (which includes the influence
  of latent variables) and the connectivity matrix $\B$ are identical across all settings $j \in \J$. In each setting $j \in \J$, the intervention shift $\mb c_j$ and the noise $\mb e_j$ are uncorrelated.
\item Interventions at different variables
  in the same setting are uncorrelated, that is $\Scj $ is an
  (unknown) diagonal matrix for all $j\in \J$.
\end{enumerate}
We will discuss a stricter version of (C) in
Section~\ref{sec:Gram} in the Appendix that allows the use of Gram matrices instead of covariance matrices. The conditions above imply that the
environments are characterized by different interventions strength, as
measured by the variance of the shift $\mb c$ in each setting. We aim
to reconstruct both the connectivity matrix $\B$ from observations in
different environments and also aim to reconstruct the a-priori unknown intervention
strength and location in each environment. Additionally, we will show examples where we can detect violations of the model assumptions and use these to reconstruct the location of interventions. 
\vspace{-0.1cm}
\subsection{Population method}
\vspace{-0.1cm}
The main idea is very simple. Looking at the
model~\eqref{eq:modelinterv}, we can rewrite 
\begin{equation}\label{eq:modelinterv2}  
(\mb I - \B) \x_j =  \mb c_j + \mb e_j. 
\end{equation}
The population covariance  of the transformed
observations are then for all settings
$j\in \J$ given by 
\begin{align}
(\mb I - \B) \Si_{\x,j} (\mb I - \B)^T &= \Scj + \Se . \label{eq:Si} \end{align}
The last term $\Se$ is constant across all settings $j\in \J$ (but not necessarily diagonal as we allow hidden variables). Any
change of the matrix on the left-hand side thus stems from a shift in the
covariance matrix $\Scj$ of the interventions. Let us define the difference
between the covariance of~$\mb c$ and $\mb x$ in setting $j$ 
as
\begin{align}\label{eq:Delta}
\DScj :=  \Scj - \Scjm,\;\;\mbox{    and     }\;\;
\DSi_{\x,j} :=  \Si_{\x,j} -  \Si_{\x,-j}.
\end{align}
Assumption (B) together with~\eqref{eq:Si} implies that
\begin{equation}\label{eq:diff}
(\mb I - \B) \DSi_{\x,j} (\mb I - \B)^T = \DScj \qquad \forall j\in \J.
\end{equation}
Using assumption (C), the random intervention shifts at different variables are
uncorrelated and the right-hand side in~\eqref{eq:diff} is thus a
diagonal matrix for all $j\in \J$. 
Let $\mathcal{D}\subset \mathbb{R}^{p\times p}$ be the set of all invertible
matrices. We also define a more restricted
space $\mathcal{D}_{cp}$ which only includes those members of $\mathcal{D}$ that have entries
all equal to one on the diagonal and have a
cycle-product less than one,
\begin{align} \label{eq:D}
\mathcal{D} & := \Big\{ \mb D \in \mathbb{R}^{p\times p}: \D \mbox{
  invertible} \Big\}\\
\mathcal{D}_{cp} & := \Big\{ \mb D \in \mathbb{R}^{p\times p}: \D \in
\mathcal{D}   \mbox{ and }   \mbox{diag}(\mb D) \equiv 1 \mbox{ and }  \CP(\mb
I-\D)<1    \Big\}. \label{eq:Dcp}
\end{align}
Under Assumption (A), 
$ \mb I - \B \in  \mathcal{D}_{cp} $.
Motivated by~\eqref{eq:diff}, we now consider the minimizer
\begin{equation} \label{eq:Dhatpop} 
\D = \mbox{argmin}_{\mb D' \in
    \mathcal{D}_{cp} } \sum_{j\in \J} L(\D ' \DSi_{\x,j} \mb D'^T ),\quad \mbox{ where } L(\mb A) := \sum_{k\neq l} \mb A_{k,l}^2
\end{equation}
is the loss $L$ for any matrix $\mb A$ and defined as the sum of the squared off-diagonal elements.
In Section~\ref{sec:identifiability}, we present necessary and sufficient conditions
on the interventions under which $\D = \mb I-\mb B$ is the unique minimizer of~\eqref{eq:Dhatpop}. 
In this case, exact
joint 
diagonalization is possible so that $L(\D \DSi_{\x,j} \mb D^T
)=0$ for all environments $j\in \J$. We discuss an alternative that replaces covariance with Gram matrices throughout in Section~\ref{sec:Gram} in the Appendix. 
We now give a finite-sample version.

\vspace{-0.1cm}
\subsection{Finite-sample estimate of the connectivity matrix}
\vspace{-0.1cm}

\begin{wrapfigure}[9]{!t}{0.55\textwidth}
\vspace{-.85cm}
\begin{minipage}{0.5\textwidth}
  \begin{algorithm}[H]
\caption{\ice \label{alg:hiddenice}}
	\algorithmicrequire\;  $\X_j \; \forall j \in \J$ 
  \begin{algorithmic}[1]
  \STATE Compute $\widehat \DSi_{\x,j} \;  \forall j \in \J$
  \STATE $\tilde{\D} = $ { \sc FFDiag}$(\widehat\DSi_{\x,j}) $ 
  \STATE $\hat{\D} = \texttt{PermuteAndScale}(\tilde{\D})$ 
  \STATE$ \hat{\B} = \mb I - \hat{\D}$
  \end{algorithmic}
  \algorithmicensure\; $\hat{\B}$
\end{algorithm} 
\end{minipage}
\vspace{-0.5cm}
\end{wrapfigure}


In practice, we estimate $\B$ by minimizing the empirical counterpart of~\eqref{eq:Dhatpop} in two steps. First, the solution of the optimization is only constrained to matrices in $\mathcal{D}$. Subsequently, we enforce the constraint on the solution to be a member of $\mathcal{D}_{cp}$. The \ice algorithm is presented in Algorithm \ref{alg:hiddenice} and we describe the important steps in more detail below.

\paragraph{Steps 1 \& 2.}
First, we minimize the following empirical, less constrained variant of~\eqref{eq:Dhatpop} 
\begin{align}\label{eq:hatD}  
\tilde{\D} &:= \mbox{argmin}_{\D'\in\mathcal{D} } \sum_{j\in \J} \; L(\D'(\widehat \DSi_{\x,j}) \mb D'^T ),
\end{align}
where the population differences between covariance matrices are replaced with their empirical counterparts and the only constraint on the solution is that it is invertible, i.e.  $\tilde{\D} \in \mathcal{D}$. For the optimization we use the joint approximate matrix diagonalization algorithm \ffdiag \cite{Ziehe04afast}. 

\paragraph{Step 3.}
The constraint on the cycle product and the diagonal elements of $\D$ is enforced by (a) permuting and (b) scaling the rows of $\tilde{\D}$. Part (b) simply scales the rows so that the diagonal elements of the resulting matrix $\hat{\D}$ are all equal to one. The more challenging first step (a) consists of finding a permutation such that 
under this permutation 
the scaled matrix from part (b) will have a cycle product as small as possible (as follows from Theorem~\ref{theorem:algo}, at most one permutation can lead to a cycle product less than one). This optimization problem seems computationally challenging at first, but we show that it can be  solved by a variant of the \emph{linear assignment problem} (LAP) (see e.g. \cite{lasp2013}), as proven in Theorem~\ref{theorem:algo} in the Appendix. As a last step, we check whether the cycle product of $\hat{\D}$ is less than one, in which case we have found the solution. Otherwise, no solution satisfying the model assumptions exists and we return a warning that the model assumptions are not met. See Appendix~\ref{sec:supp_algo} for more details.

\paragraph{Computational cost.} 
The computational complexity of \ice is  $O(\vert \J \vert \cdot n \cdot p^2)$ as
computing the covariance matrices costs $O(\vert \J \vert \cdot n \cdot p^2)$, \ffdiag has a computational cost of $O(\vert \J \vert \cdot p^2)$ and both the linear assignment problem and computing the cycle product can be solved in $O(p^3)$ time. For instance, this complexity is achieved when using the Hungarian algorithm for the linear assignment problem (see e.g. \cite{lasp2013}) and the cycle product can be computed with a simple dynamic programming approach.


\vspace{-0.1cm}
\subsection{Estimating the intervention variances}\label{sec:Int_vars}
\vspace{-0.1cm}
One additional benefit of \ice is that the location and strength of the interventions can be estimated from the data. The empirical, plug-in version of Eq.~\eqref{eq:diff} is given by
\begin{equation}\label{eq:diff_emp}
(\mb I - \hat{\B}) \widehat{\DSi}_{\x,j} (\mb I - \hat{\B})^T = \widehat{\DSi}_{\mb c,j}  = \widehat{\Si}_{\mb c,j} - \widehat{\Si}_{\mb c,-j}\qquad \forall j\in \J.
\end{equation}
So the element $( \widehat{\DSi}_{\mb c,j} )_{kk}$ is an estimate for the difference between the variance of the intervention at variable $k$ in environment $j$, namely $( \Scj )_{kk}$, and the average in all other environments, $( \Scjm )_{kk}$. From these differences we can compute the intervention variance for all environments up to an offset. By convention, we set the minimal intervention variance across all environments equal to zero. Alternatively, one can let observational data, if available, serve as a baseline against which the intervention variances are measured. 

\vspace{-0.1cm}
\section{Identifiability}\label{sec:identifiability}
\vspace{-0.1cm}
Let for simplicity of notation,
\[ \sic_{j,k} := ( \DScj )_{kk} \]
be the variance of the random intervention shifts $\mb c_j$
at node $k$ in environment $j\in \J$ as per the definition of $\DScj$ in \eqref{eq:Delta}. We then have the following identifiability result (the proof is provided in Appendix~\ref{sec:supp_ident}).
\begin{theorem}\label{theorem:1}
	Under assumptions (A), (B) and (C), the solution to \eqref{eq:Dhatpop} is unique 
	if and only if for all $k,l \in \left\{ 1,\dots,p \right\}$ there exist $j,j' \in \J$ such that
\begin{equation}\label{eq:manyunique} \sic_{j,k} \sic_{j',l} \neq \sic_{j,l}  \sic_{j',k}\,. \end{equation}
\end{theorem}
If none of the intervention variances $ \sic_{j,k} $ vanishes, the uniqueness condition is equivalent to demanding that the ratio between the intervention variances for two variables $k,l$ must not stay identical across all environments, that is there exist $j,j'\in \J$ such that
\begin{equation}\label{eq:manyuniquenon0} \frac{\sic_{j,k} }{\sic_{j,l} } \neq \frac{\sic_{j',k}}{\sic_{j',l}}  ,  \end{equation}
which requires that the ratio of the  variance of the intervention shifts at two nodes $k,l$ is not identical across all settings.
This leads to the following corollary.
\begin{cor} 
\begin{enumerate}[(i)]
\item The identifiability condition~\eqref{eq:manyunique} cannot be satisfied if $|\J|=2$ 
since then $\sic_{j,k} = - \sic_{j',k}$ for all $k$ and $j \neq j'$. 
We need at least three different environments for identifiability.
\item The identifiability condition~\eqref{eq:manyunique} is satisfied for all $|\J|\ge 3$ almost surely if the variances of the intervention $\mb c_j$ are chosen independently 
(over all variables and environments $j \in \J$) 
from a distribution that is absolutely continuous with respect to Lebesgue measure.
\end{enumerate}

\end{cor}
Condition (ii) can be relaxed but shows that we can already  achieve full identifiability with a very generic setting for three (or more) different environments.

\vspace{-0.1cm}
\section{Numerical results}\label{sec:numerical}
\vspace{-0.1cm}
In this section, we present empirical results for both synthetic and real data sets. In addition to estimating the connectivity matrix $\B$, we demonstrate various ways to estimate properties of the interventions.
Besides computing the point estimate for \ices, we use \emph{stability selection} \cite{meinshausen2010stability} to assess the stability of retrieved edges.
 We attach R-code  with which all simulations and analyses can be reproduced\footnote{An R-package called ``\texttt{backShift}'' is available from CRAN.}.
\subsection{Synthetic data}\label{subsec:simulated}
We compare the point estimate of \ice against \ling \cite{lacerda2012discovering}, a generalization of \lingam to the cyclic case for purely observational data. We consider the cyclic graph shown in Figure~\ref{fig:sim_true_metrics}\subref{fig:sim_net} and generate data under different scenarios. The data generating mechanism is sketched in Figure~\ref{fig:sim_true_metrics}\subref{fig:sim_gen}.
Specifically, we generate ten distinct environments with non-Gaussian noise. 
In each environment, the random intervention variable is generated as  $(\ci_j)_k =\beta_k^j I_k^j$, 
where $\beta_1^j, \ldots, \beta_p^j$ are drawn \iid from $ \text{Exp}(\iMult)$ and  $I_1^j, \ldots, I_p^j$ are independent standard normal random variables. 
The intervention shift thus acts on all observed random variables.
The parameter $\iMult$ regulates the strength of the intervention.
If hidden variables exist, the noise term $(\mb e_j)_k$ of variable $k$ in environment $j$ is equal to $ \gamma_k W^j$, where the weights $\gamma_1, \ldots, \gamma_p$ are sampled once from a  $\N(0,1)$-distribution and the random variable $W^j$ has a $ \text{Laplace}(0, 1)$ distribution.
If no hidden variables are present, then $(\mb e_j)_k$, $k=1,\ldots,p$  is sampled \iid $\text{Laplace}(0, 1)$. 
In this set of experiments, we consider five different settings (described below) in which the sample size $n$, the intervention strength $\iMult$ as well as the existence of hidden variables varies. 

\begin{figure*}[!tp]
\begin{centering}
\vspace{-.5cm}
\hspace{-.55cm}
\subfloat[\xspace]{
    \includegraphics[trim=85 80 50 70, clip, width=0.23\textwidth, keepaspectratio=true]{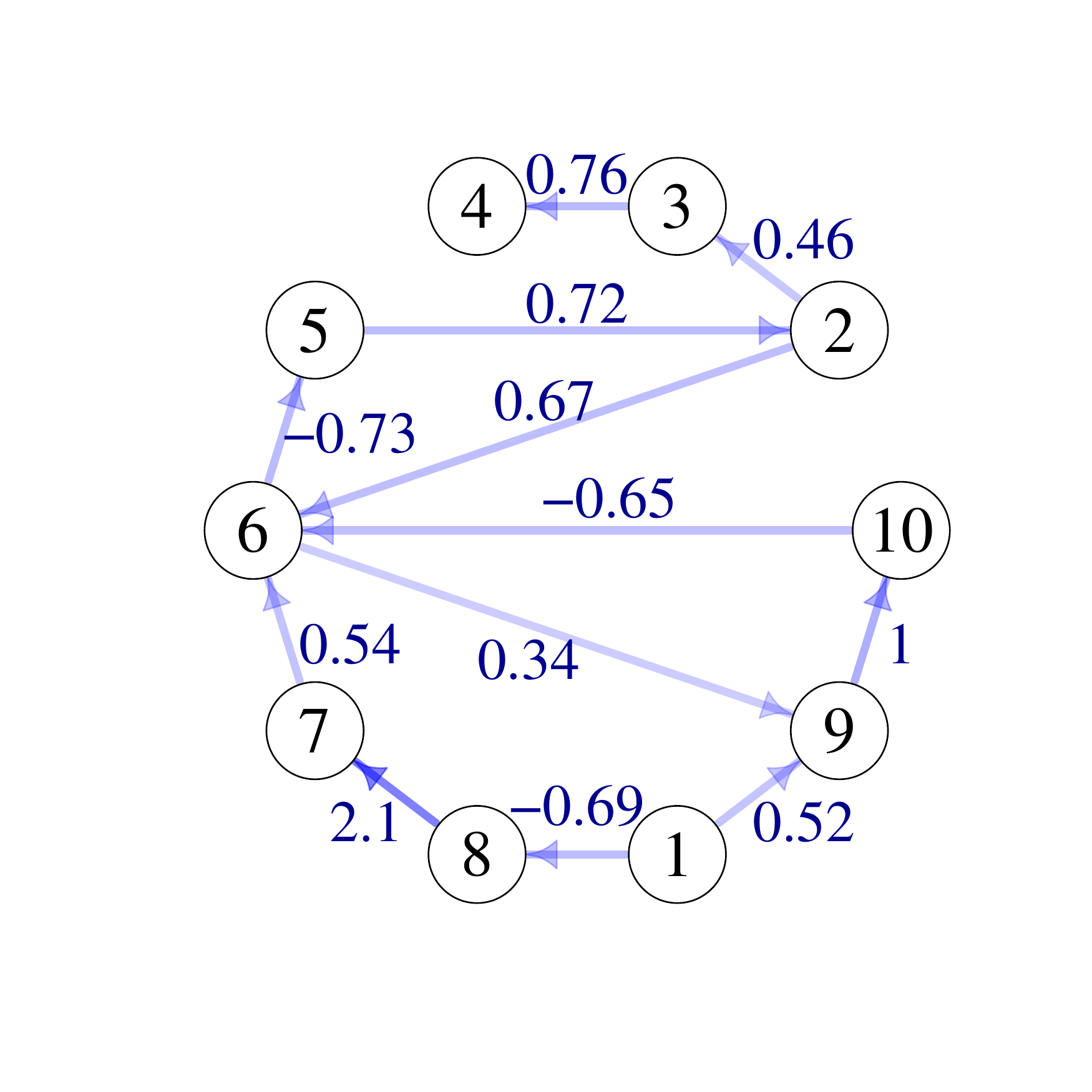}
    \label{fig:sim_net}
}
\hspace{.1cm}
\subfloat[\xspace]{
\begin{tikzpicture}[scale=.57, line width=0.5pt, minimum size=0.45cm, inner sep=0.3mm, shorten >=1pt, shorten <=1pt]
    \tikzstyle{every node}=[font=\tiny]
 
   	\draw (-1.5,0) node(x1) [circle, draw] {$X_1$};
    \draw (1.5,0) node(x2) [circle, draw] {$X_2$};
    \draw (0,2.5) node(x3) [circle, draw] {$X_3$};
   
   	\draw (-2.5,-1) node(i1) [circle, draw, blue] {$I_1$};
    \draw (2.5,-1) node(i2) [circle, draw, blue] {$I_2$};
    \draw (0.75,3.75) node(i3) [circle, draw, blue] {$I_3$};
    
	\draw (-2.5,1) node(w1) [circle, dashed, draw, magenta] {$E_1$}; 
	\draw (2.5,1) node(w2) [circle, dashed, draw, magenta] {$E_2$};  
	\draw (-.75,3.75) node(w3) [circle, dashed, draw, magenta] {$E_3$};    
	
	 \draw (0,1) node(w) [circle, dashed, draw, magenta] {$W$};   
    
    \draw (-2.4,-0.3) node(ii) [blue] {$\beta_1$};
    \draw (2.4,-0.3) node(ii) [blue] {$\beta_2$};
	\draw (.8,3) node(iii) [blue] {$\beta_3$};     
	
	\draw (-.7,.775) node(ii) [magenta] {$\gamma_1$};
	\draw (.7,.775) node(ii) [magenta] {$\gamma_2$};
	 \draw (-.3,1.5) node(ii) [magenta] {$\gamma_3$};
     
    \draw[-arcsq] (x1) -- (x2);
    \draw[-arcsq] (x2) -- (x3);
    \draw[-arcsq] (x3) -- (x1);
    \draw[-arcsq, blue] (i2) -- (x2);
    \draw[-arcsq, blue] (i3) -- (x3);
    \draw[-arcsq, blue] (i1) -- (x1);
    \draw[-arcsq, magenta] (w1) -- (x1);
    \draw[-arcsq, magenta] (w2) -- (x2);
    \draw[-arcsq, magenta] (w3) -- (x3);
	\draw[-arcsq, magenta] (w) -- (x1);
    \draw[-arcsq, magenta] (w) -- (x2);
    \draw[-arcsq, magenta] (w) -- (x3);    
    
   \end{tikzpicture}  
  
    \label{fig:sim_gen}
    }
\hspace{.2cm}
\subfloat[\xspace]{
  \includegraphics[trim=0 55 10 50, clip, width=0.35\textwidth, keepaspectratio=true]{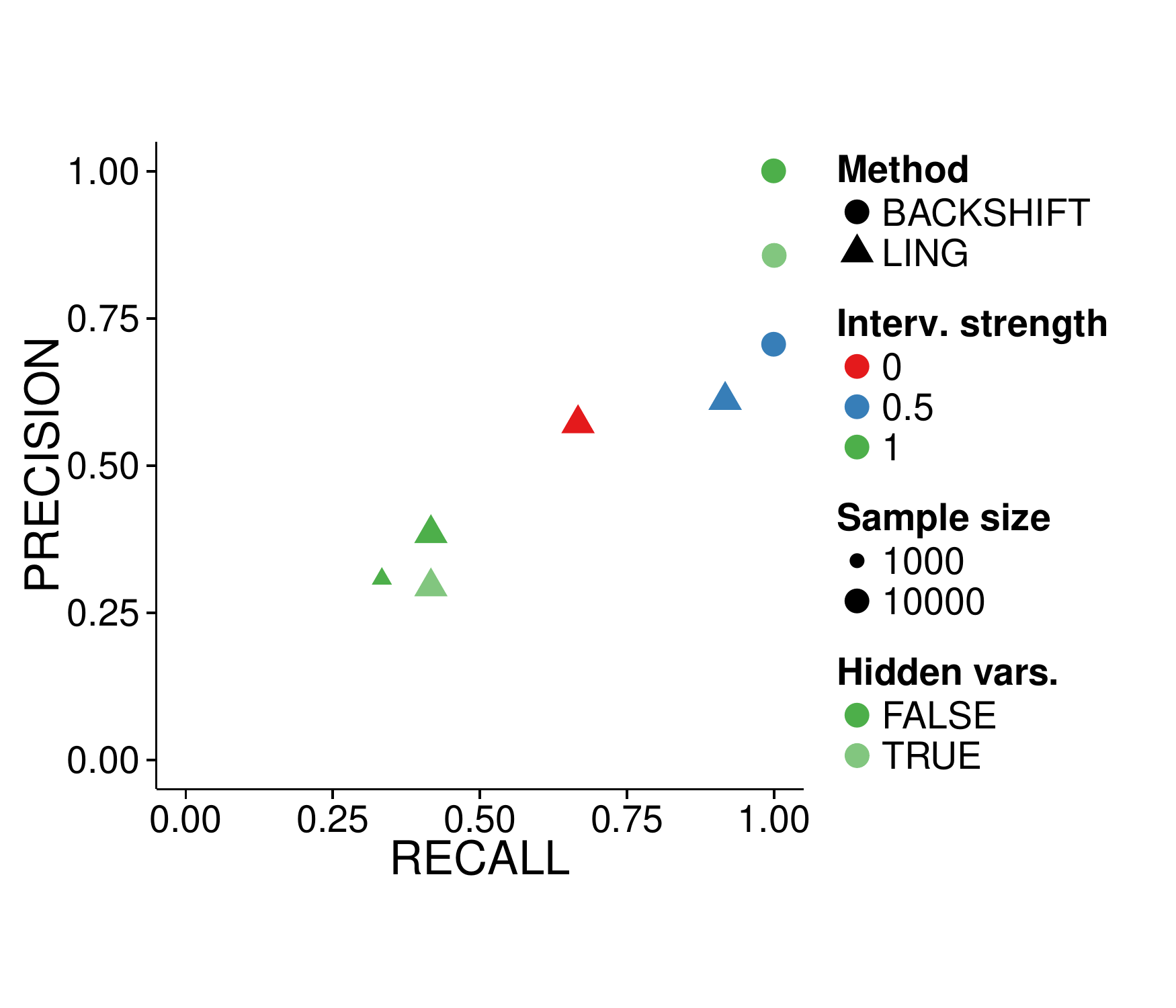}
    \label{fig:sim_metrics}

}
\vspace{-5pt}
\caption{\small Simulated data. \protect\subref{fig:sim_net} True network. 
\protect\subref{fig:sim_gen} Scheme for data generation. 
\protect\subref{fig:sim_metrics} Performance metrics for the settings considered in Section \ref{subsec:simulated}. For \ices, precision and recall values for Settings~1 and~2 coincide.} \label{fig:sim_true_metrics}
\end{centering}
\vspace{-0.2cm}
\end{figure*}
\begin{table*}[!tp]
\hspace{-0.2cm}
  \begin{tabular}
      {llllll} 
	& 
	 \hspace{1mm}     
     {\small \emph{Setting 1}} & 
     \hspace{-2mm}
     {\small \emph{Setting 2}} & 
     \hspace{-2mm}
     {\small \emph{Setting 3}} & 
     \hspace{-2mm}
     {\small \emph{Setting 4}} & 
     \hspace{-2mm}
     {\small \emph{Setting 5}} \\      
      
     & 
	 \hspace{1mm}     
     {\small $n = 1000$} & 
     \hspace{-2mm}
     {\small $n = 10000$} & 
     \hspace{-2mm}
     {\small $n = 10000$} & 
     \hspace{-2mm}
     {\small $n = 10000$} & 
     \hspace{-2mm}
     {\small $n = 10000$} \\

     & 
	 \hspace{1mm}
     {\small no hidden vars.} & 
     \hspace{-2mm}
     {\small no hidden vars.} & 
     \hspace{-2mm}
     {\small hidden vars.} & 
     \hspace{-2mm}
     {\small no hidden vars.} & 
     \hspace{-2mm}
     {\small no hidden vars.} \\
 
 	& 
 	\hspace{1mm}
 	{\small $\iMult = 1$} & 
 	 \hspace{-2mm}
     {\small $\iMult = 1$} & 
     \hspace{-2mm}
     {\small $\iMult = 1$} & 
     \hspace{-2mm}
     {\small $\iMult = 0$} & 
     \hspace{-2mm}
     {\small $\iMult = 0.5$} 
     \vspace{-0.65cm}\\

\rotatebox{90}{\ice} & 
\hspace{-3mm}
\subfloat{
  \raisebox{-.1\height}{\includegraphics[trim=85 85 50 02, clip, width=0.17\textwidth, keepaspectratio=true]{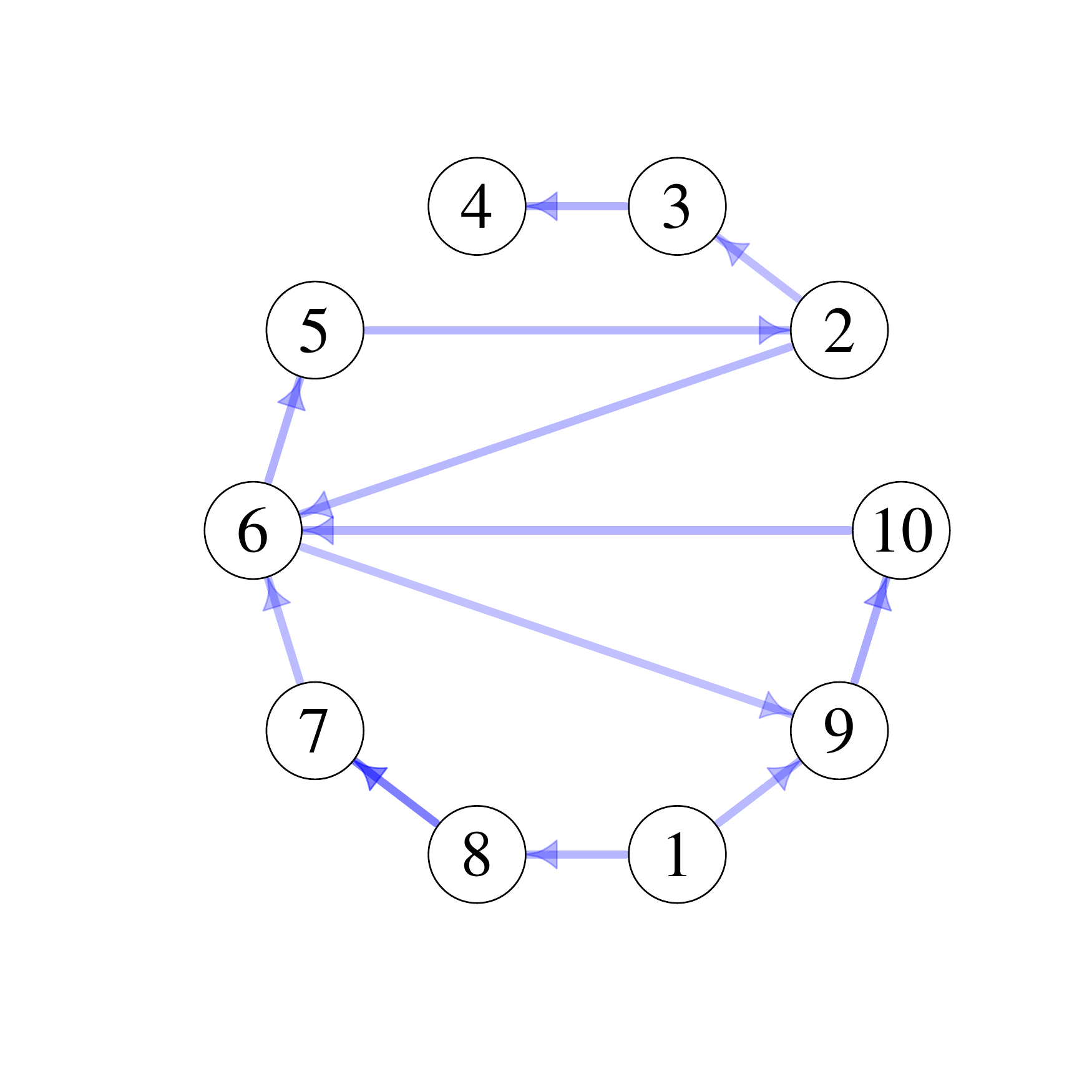}
}} & \hspace{-6mm}
\subfloat{
     \raisebox{-.1\height}{\includegraphics[trim=85 85 50 70, clip, width=0.17\textwidth, keepaspectratio=true]{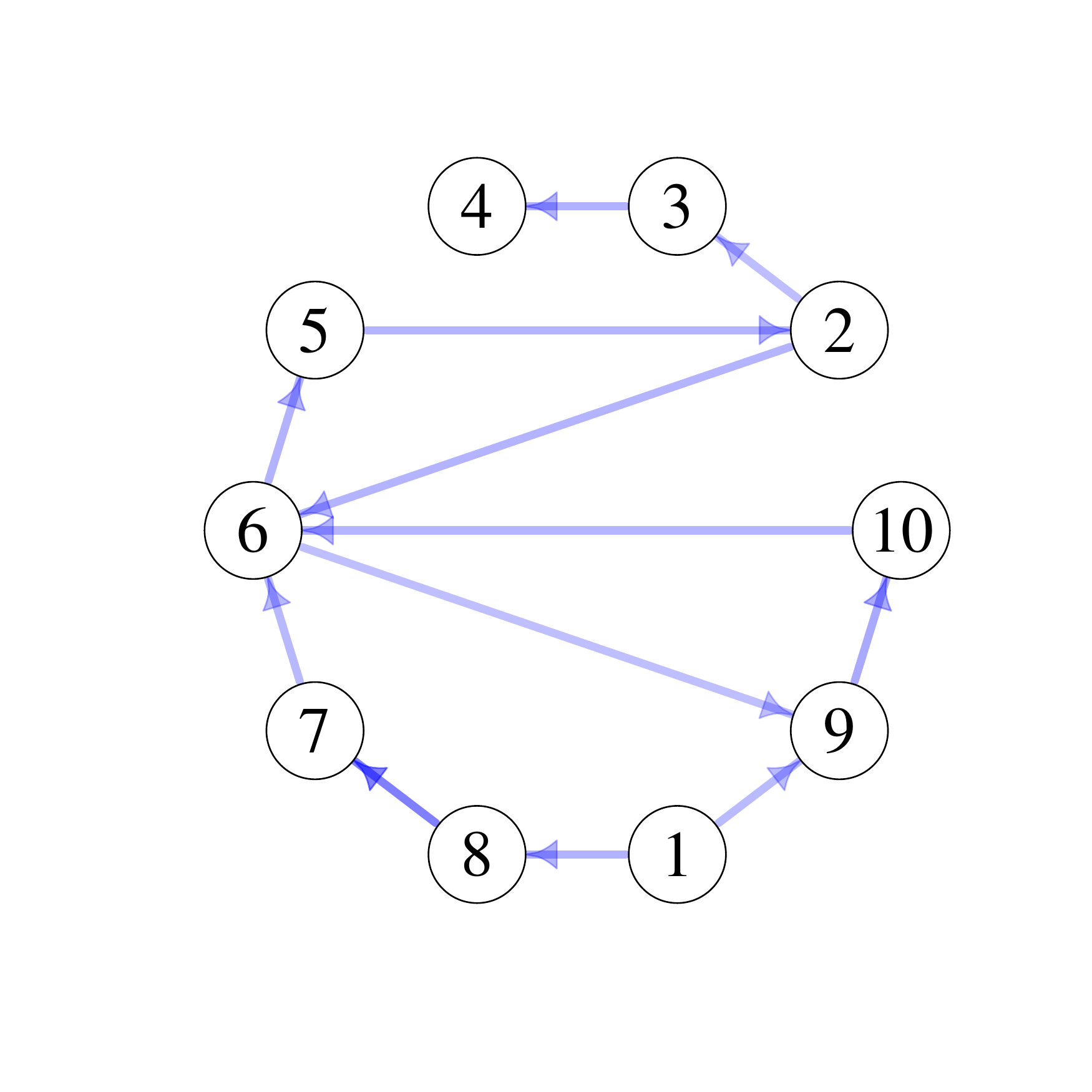}
}} & \hspace{-6mm}
\subfloat{
     \raisebox{-.1\height}{\includegraphics[trim=85 85 50 70, clip, width=0.17\textwidth, keepaspectratio=true]{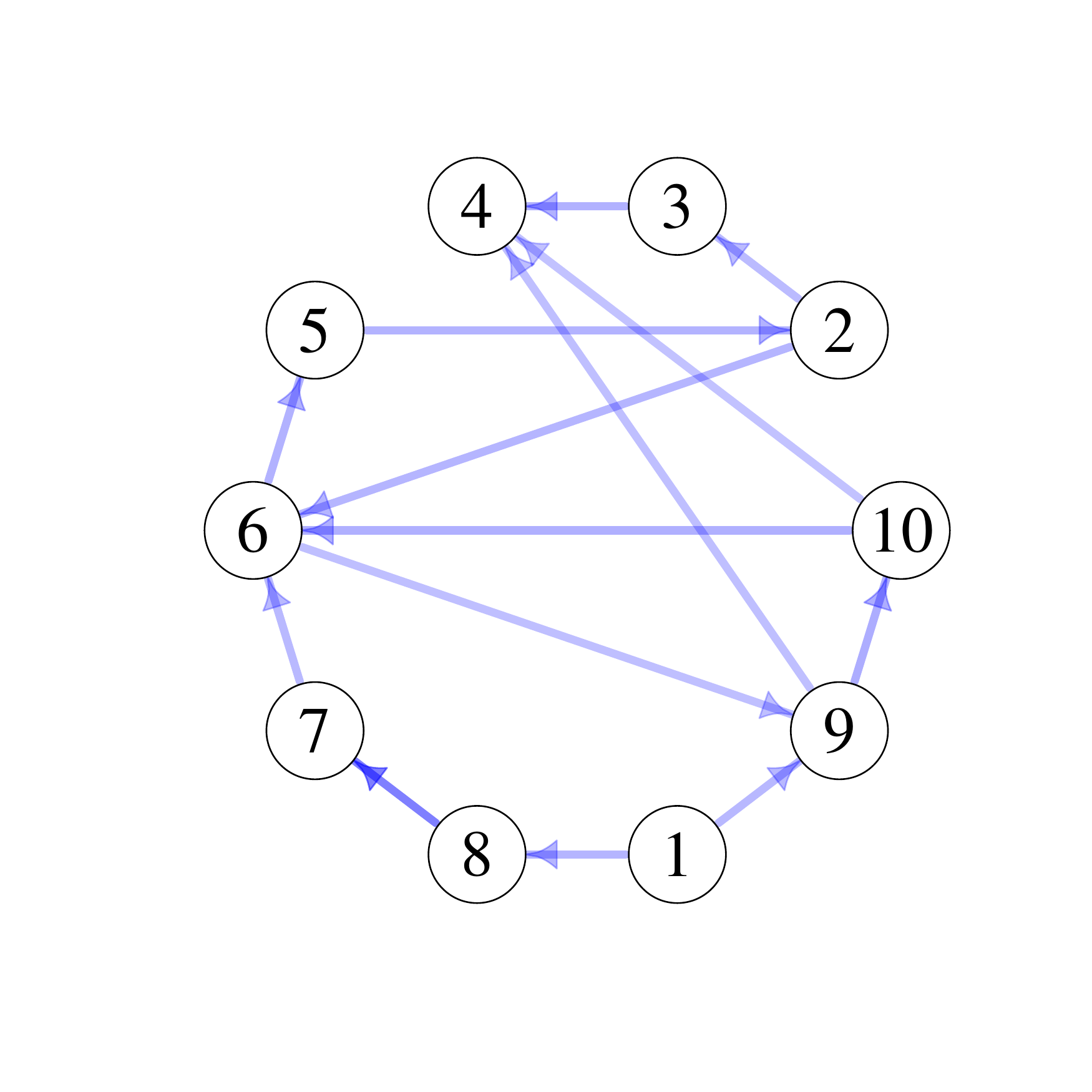}
}} & \hspace{-6mm}
\subfloat{
     \raisebox{-.1\height}{\includegraphics[trim=85 85 50 70, clip, width=0.17\textwidth, keepaspectratio=true]{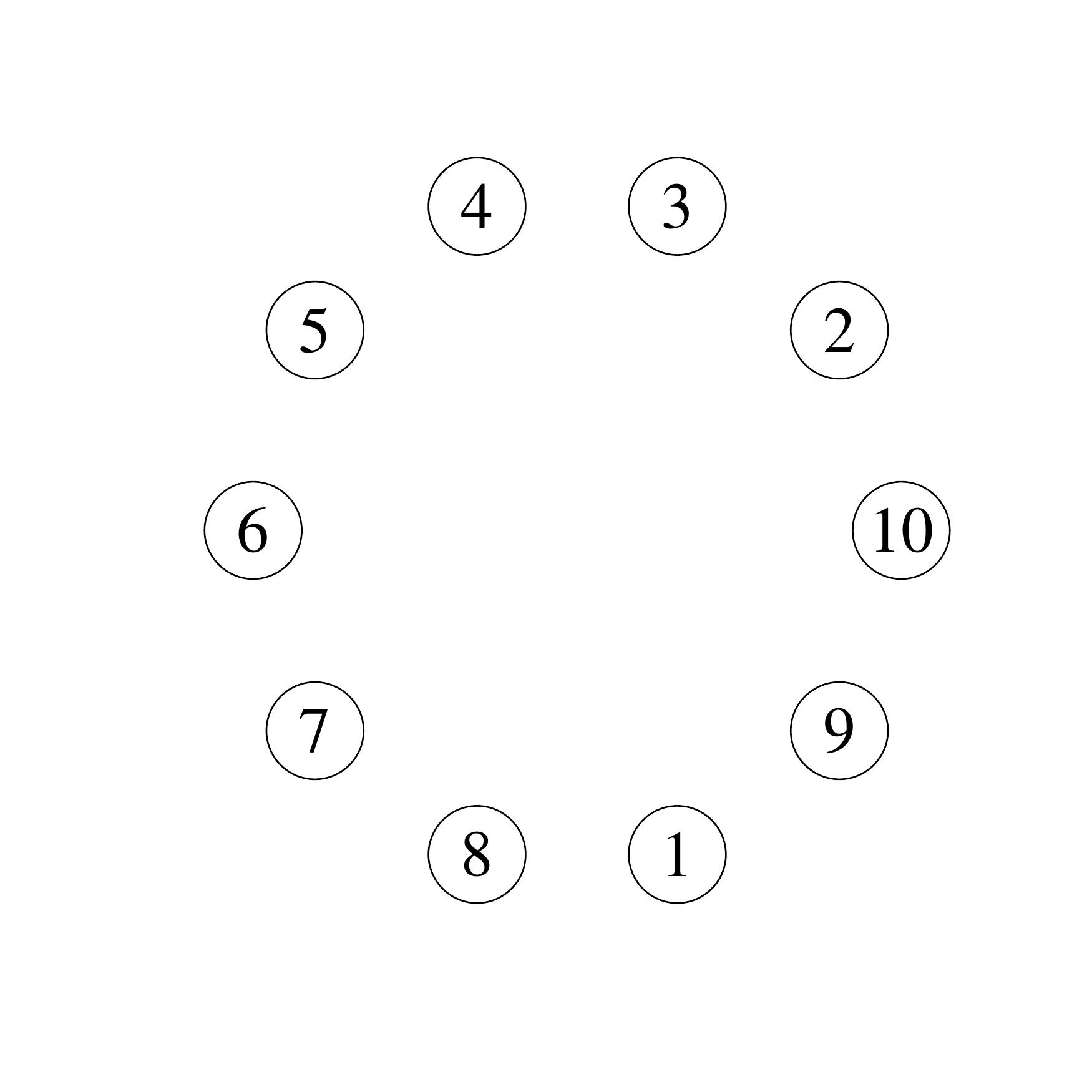}
}} & \hspace{-6mm}
\subfloat{
    \raisebox{-.1\height}{\includegraphics[trim=85 85 50 70, clip, width=0.17\textwidth, keepaspectratio=true]{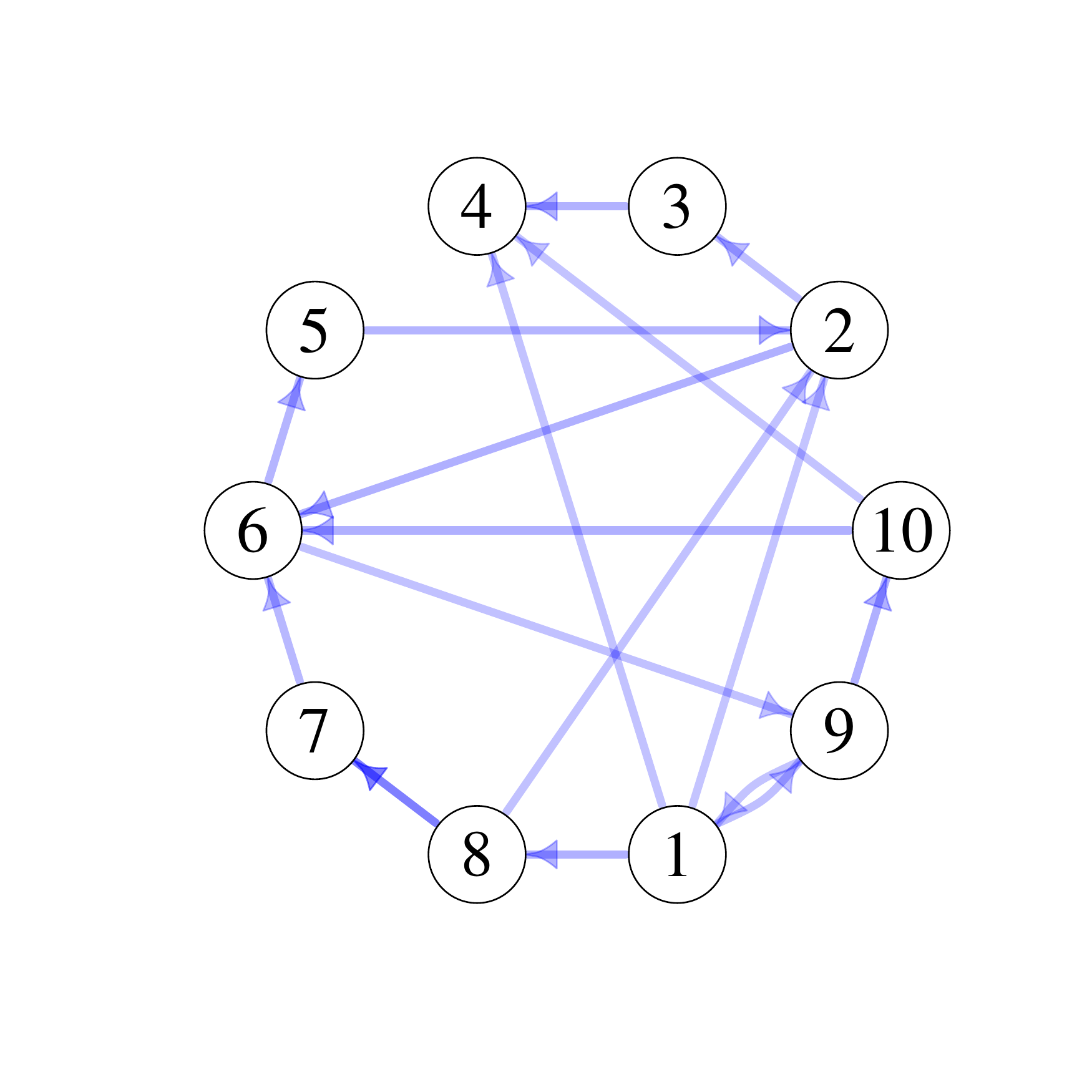}
}}
 \\
  & \hspace{-2mm}  \tiny ${\bf \mathbf{SHD}= 0}, \vert t \vert = 0.25$ &\hspace{-2mm}  \tiny${\bf \mathbf{SHD} = 0}, \vert t \vert = 0.25$&\hspace{-2mm}  \tiny${\bf \mathbf{SHD} = 2}, \vert t \vert = 0.25$ &\hspace{-2mm}  \tiny${\bf \mathbf{SHD} = 12}$&\hspace{-2mm}  \tiny${\bf \mathbf{SHD} = 5}, \vert t \vert = 0.25$
  \vspace{-0.15cm}\\
\rotatebox{90}{\ling} &
\hspace{-3mm}
\subfloat{
     \raisebox{-.4\height}{\includegraphics[trim=85 85 50 70, clip, width=0.17\textwidth, keepaspectratio=true]{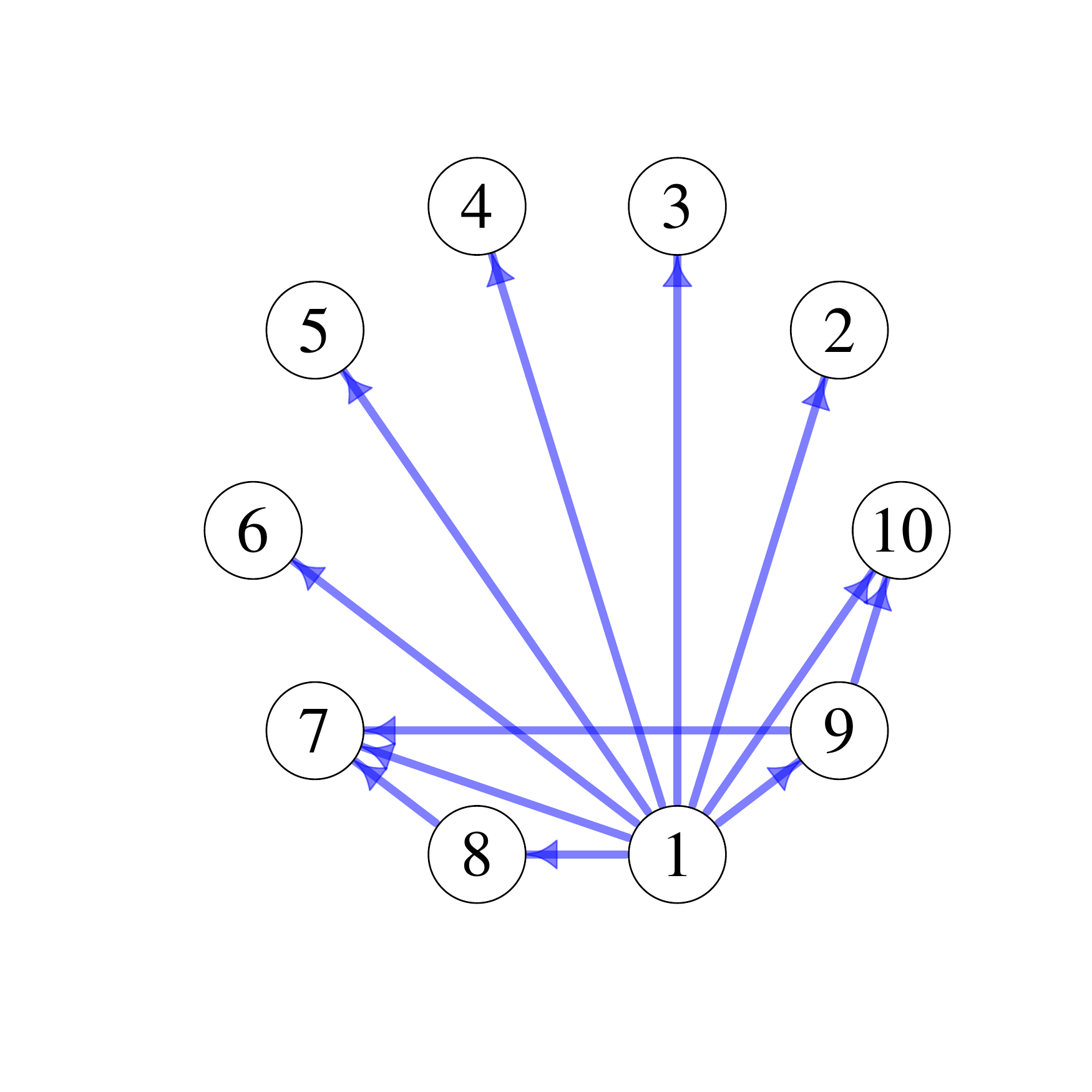}
}} & \hspace{-6mm}
\subfloat{
    \raisebox{-.4\height}{\includegraphics[trim=85 85 50 70, clip, width=0.17\textwidth, keepaspectratio=true]{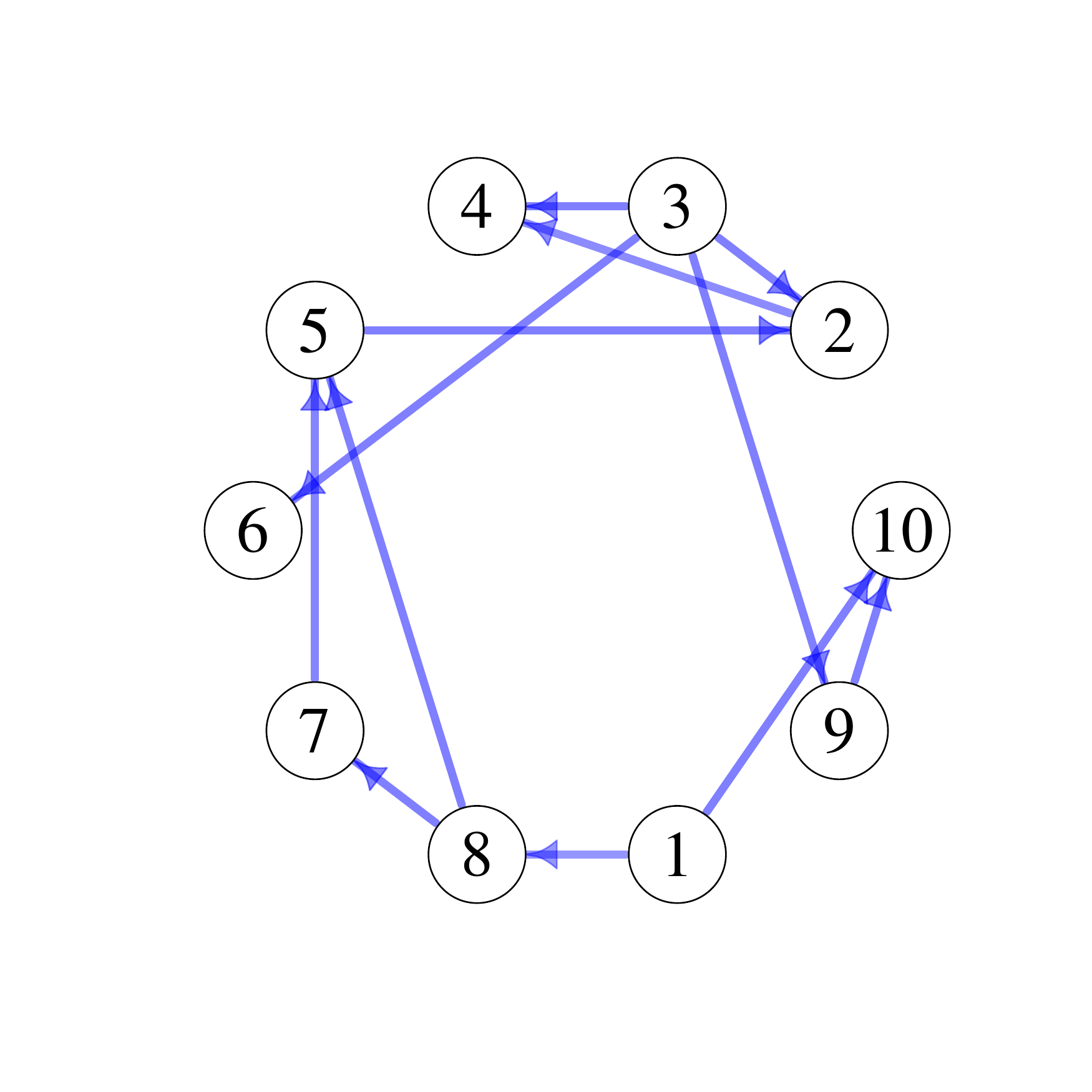}
}} & \hspace{-6mm}
\subfloat{
    \raisebox{-.4\height}{\includegraphics[trim=85 85 50 70, clip, width=0.17\textwidth, keepaspectratio=true]{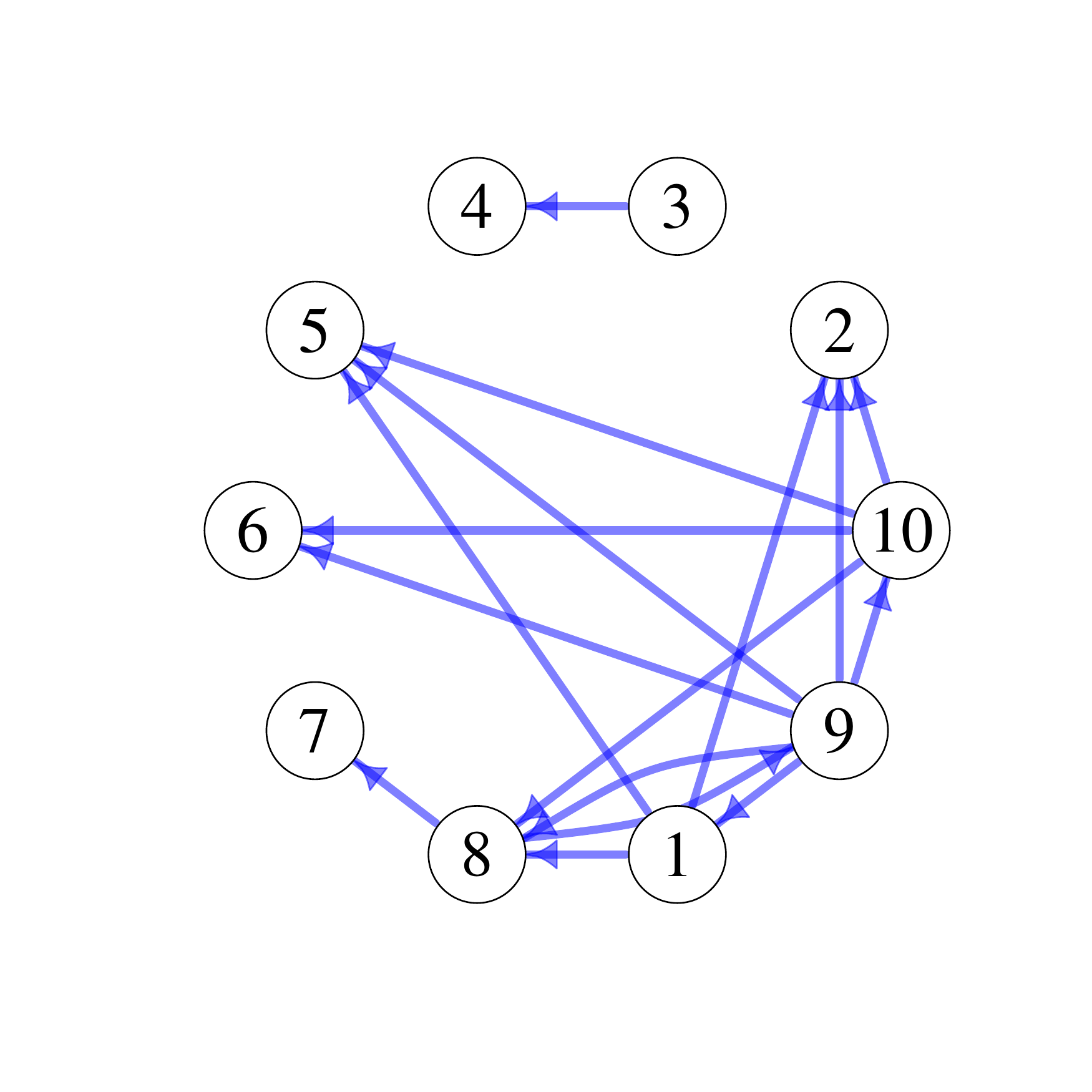}
}} & \hspace{-6mm}
\subfloat{
    \raisebox{-.4\height}{\includegraphics[trim=85 85 50 70, clip, width=0.17\textwidth, keepaspectratio=true]{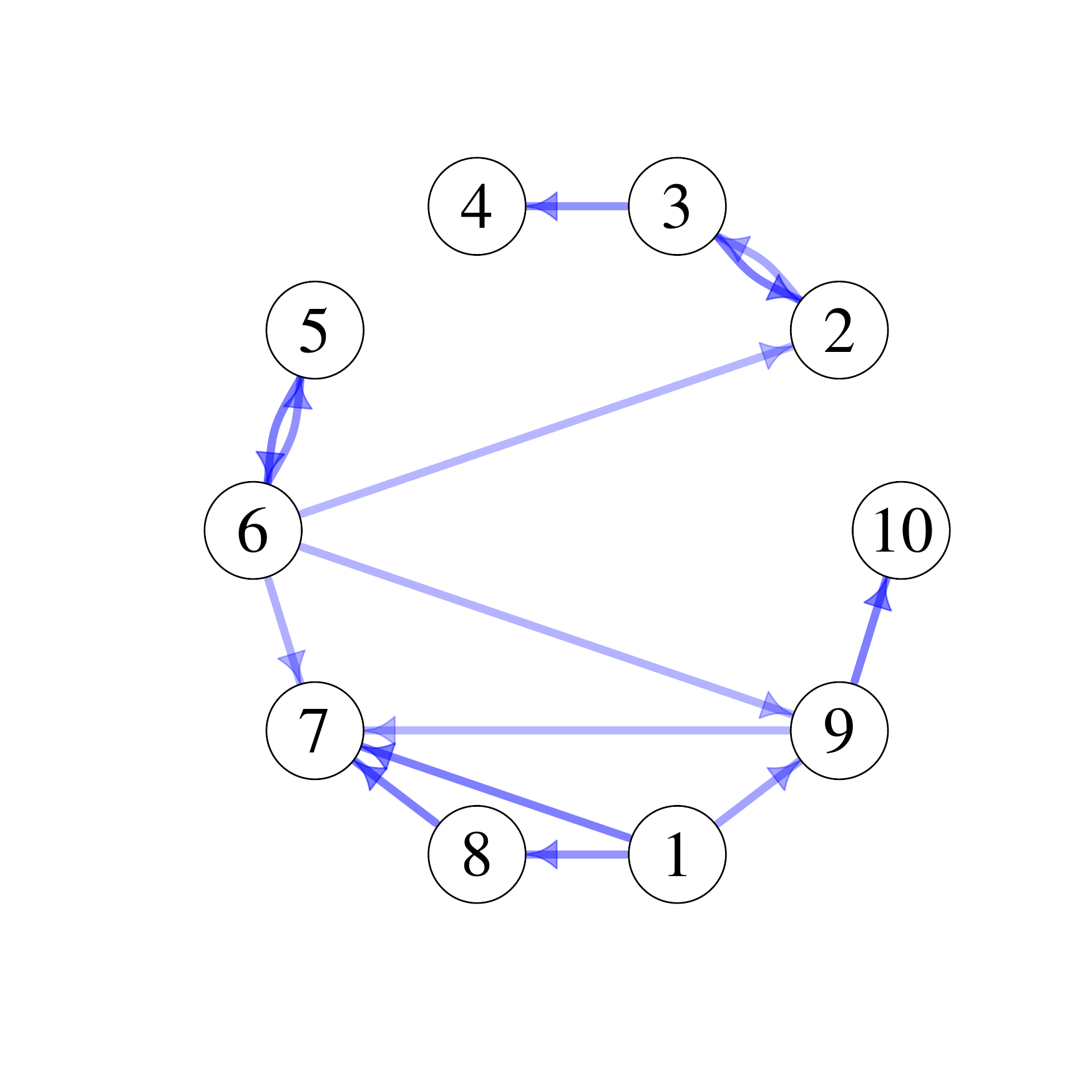}
}} & \hspace{-6mm}
\subfloat{
    \raisebox{-.4\height}{\includegraphics[trim=85 85 50 70, clip, width=0.175\textwidth, keepaspectratio=true]{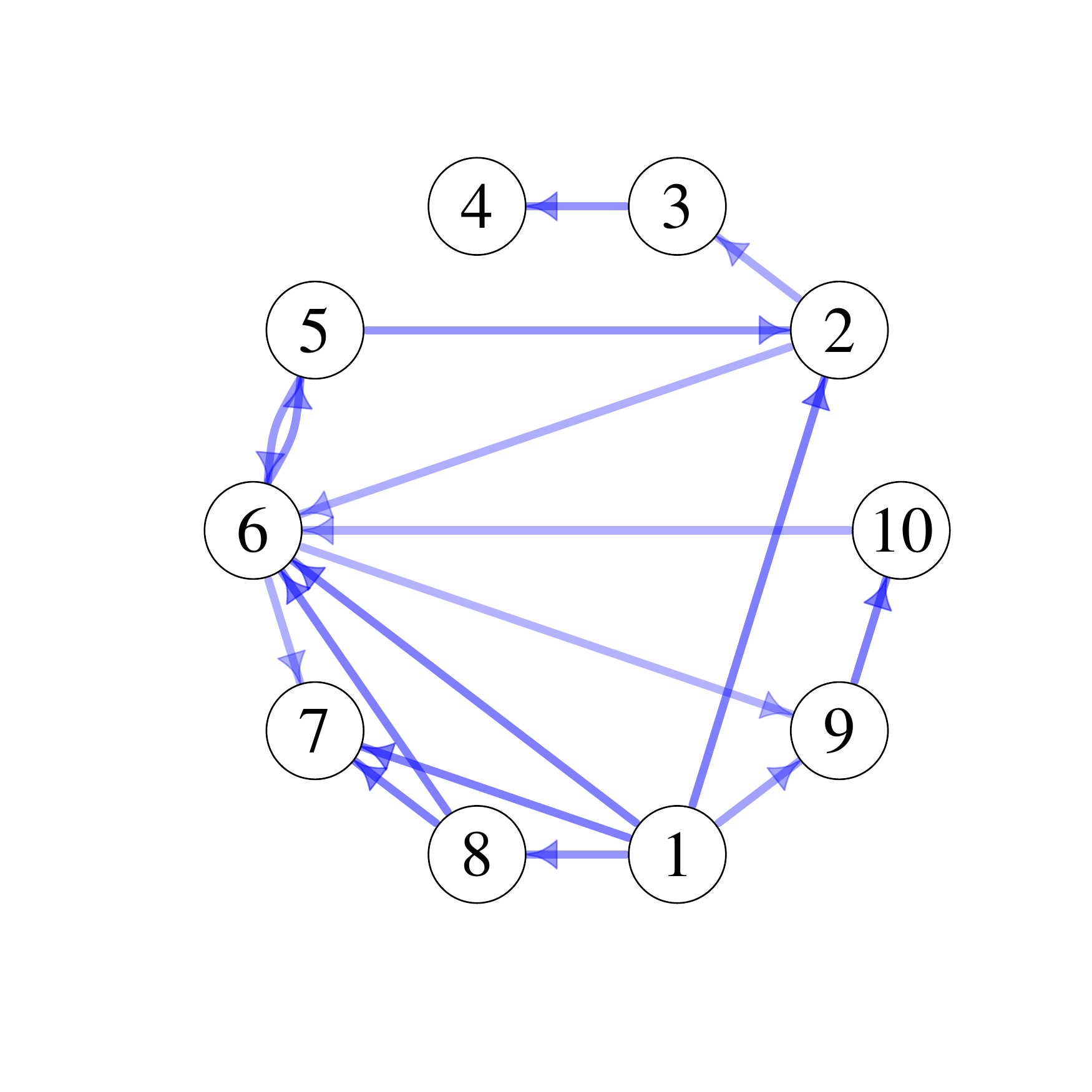}
}} \\
&  \hspace{-2mm} \tiny ${\bf \mathbf{SHD}= 17}, \vert t \vert = 0.91$ &  \hspace{-2mm} \tiny${\bf \mathbf{SHD} = 14}, \vert t \vert = 0.68$&  \hspace{-2mm}\tiny${\bf \mathbf{SHD} = 16}, \vert t \vert = 0.98$ &  \hspace{-2mm} \tiny${\bf \mathbf{SHD} = 8}, \vert t \vert = 0.25$&  \hspace{-2mm} \tiny${\bf \mathbf{SHD} = 7}, \vert t \vert = 0.29$\\
  \end{tabular}
 \captionof{figure}{{\small Point estimates of \ice and \ling for synthetic data. We threshold the point estimate of \ice at $t = \pm 0.25$ to exclude those entries which are close to zero. We then threshold the estimate of \ling so that the two estimates have the same number of edges. 
 In Setting~4, we threshold \ling at $t = \pm 0.25$ as \ice returns the empty graph. In Setting~3, it is not possible to achieve the same number of edges as all remaining coefficients in the point estimate of \ling are equal to one in absolute value. The transparency of the edges illustrates the relative magnitude of the estimated coefficients. We report the structural Hamming distance ($\mathrm{SHD}$) for each graph. Precision and recall values are shown in Figure~\ref{fig:sim_true_metrics}\protect\subref{fig:sim_metrics}.
 }} \label{fig:sim_results1}
\vspace{-5.5mm}
\end{table*}

We allow for hidden variables in only one out of five settings as \ling assumes causal sufficiency and can thus in theory not cope with hidden variables.
If no hidden variables are present, the pooled data can be interpreted as coming from a model whose error variables follow a mixture distribution. But if one of the error variables comes from the second mixture component, for example, the other error variables come from the second mixture component, too. In this sense, the data points are not independent anymore. This poses a challenge for \ling which assumes an \iid sample.
We also cover a case (for $\iMult=0$) in which all assumptions of \ling are satisfied (Scenario 4). 

Figure~\ref{fig:sim_results1} shows the estimated connectivity
matrices for five different settings and
Figure~\ref{fig:sim_true_metrics}\subref{fig:sim_metrics} shows the
obtained precision and recall values. In Setting~1, $n = 1000$,
$\iMult = 1$ and there are no hidden variables. In Setting~2, $n$ is
increased to $10000$ while the other parameters do not change. We
observe that \ice retrieves the correct adjacency matrix in both cases
while \lings's estimate is not very accurate. It improves slightly
when increasing the sample size. 
In Setting~3, we do include hidden variables which violates the causal sufficiency assumption required for \lings. Indeed, the estimate is worse than in Setting~2 but somewhat better than in Setting~1. \ice retrieves two false positives in this case. 
Setting~4 is not feasible for \ice as the distribution of the variables is identical in all environments (since $\iMult=0$). In Step 2 of the algorithm, \ffdiag does not converge and therefore the empty graph is returned. So the recall value is zero while precision is not defined. For \ling all assumptions are satisfied and the estimate is more accurate than in the Settings~1--3. Lastly, Setting~5 shows that when increasing the intervention strength to $0.5$, \ice returns a few false positives. Its performance is then similar to \ling which returns its most accurate estimate in this scenario.
The stability selection results for \ice are provided in Figure~\ref{fig:sim_results_stabSel} in Appendix~\ref{sec:appFig5}.

In short, these results suggest that the \ice point estimates are close to the true graph if the interventions are sufficiently strong. Hidden variables make the estimation problem more difficult but the true graph is recovered if the  strength of the intervention is increased (when increasing $\iMult$ to $1.5$ in Setting~3, \ice obtains a $\mathrm{SHD}$ of zero). In contrast, \ling is unable to cope with hidden variables but also has worse accuracy in the absence of hidden variables under these shift interventions.

%
%
%
\begin{figure*}[!tp]
\begin{centering}
\vspace{-.1cm}
\hspace{-.65cm}
\subfloat[]
{
    \includegraphics[trim=50 80 50 70, clip, width=0.23\textwidth, keepaspectratio=true]{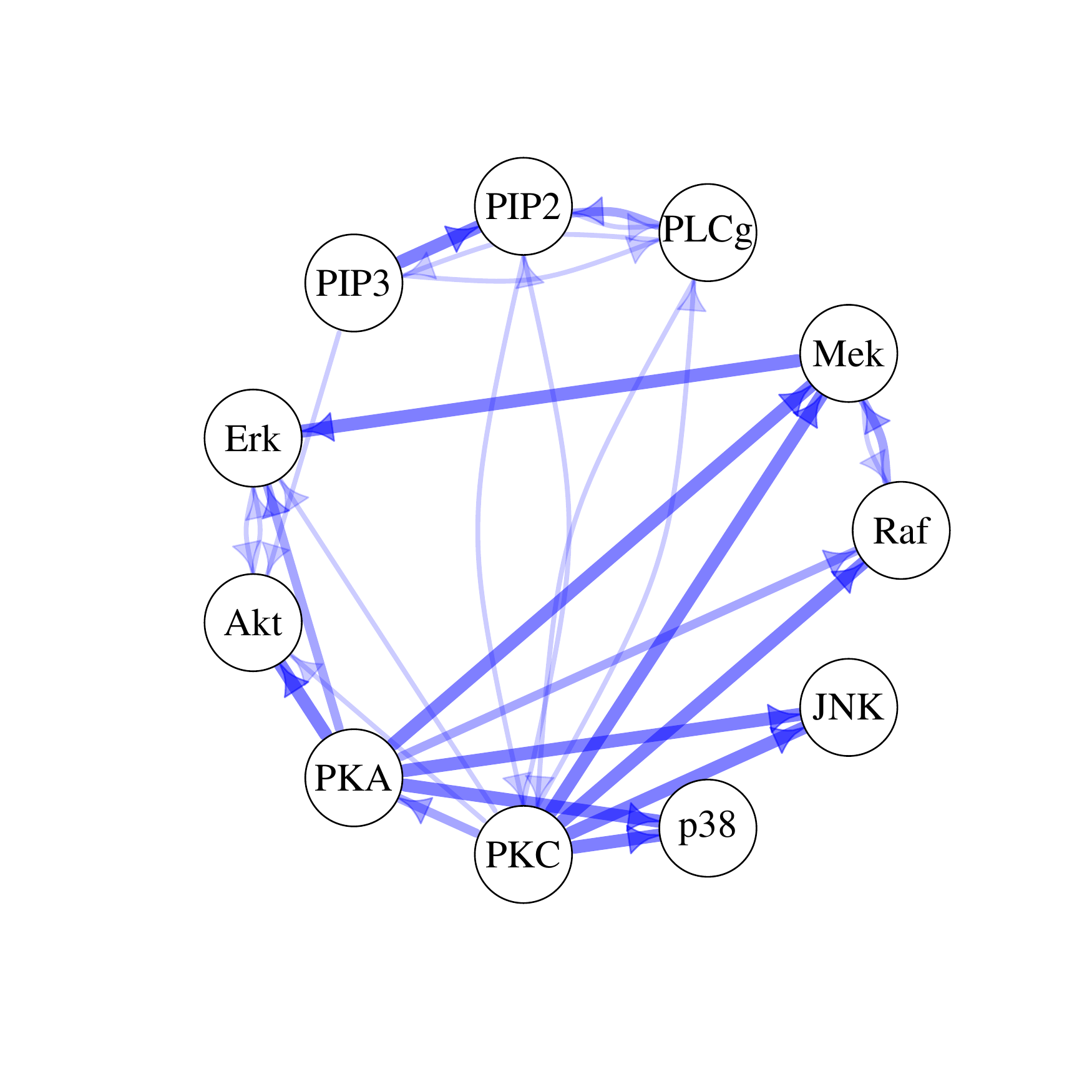}
    \label{fig:cons}
}
\hspace{-0.15cm}
\subfloat[]
{
    \includegraphics[trim=50 80 50 70, clip, width=0.23\textwidth, keepaspectratio=true]{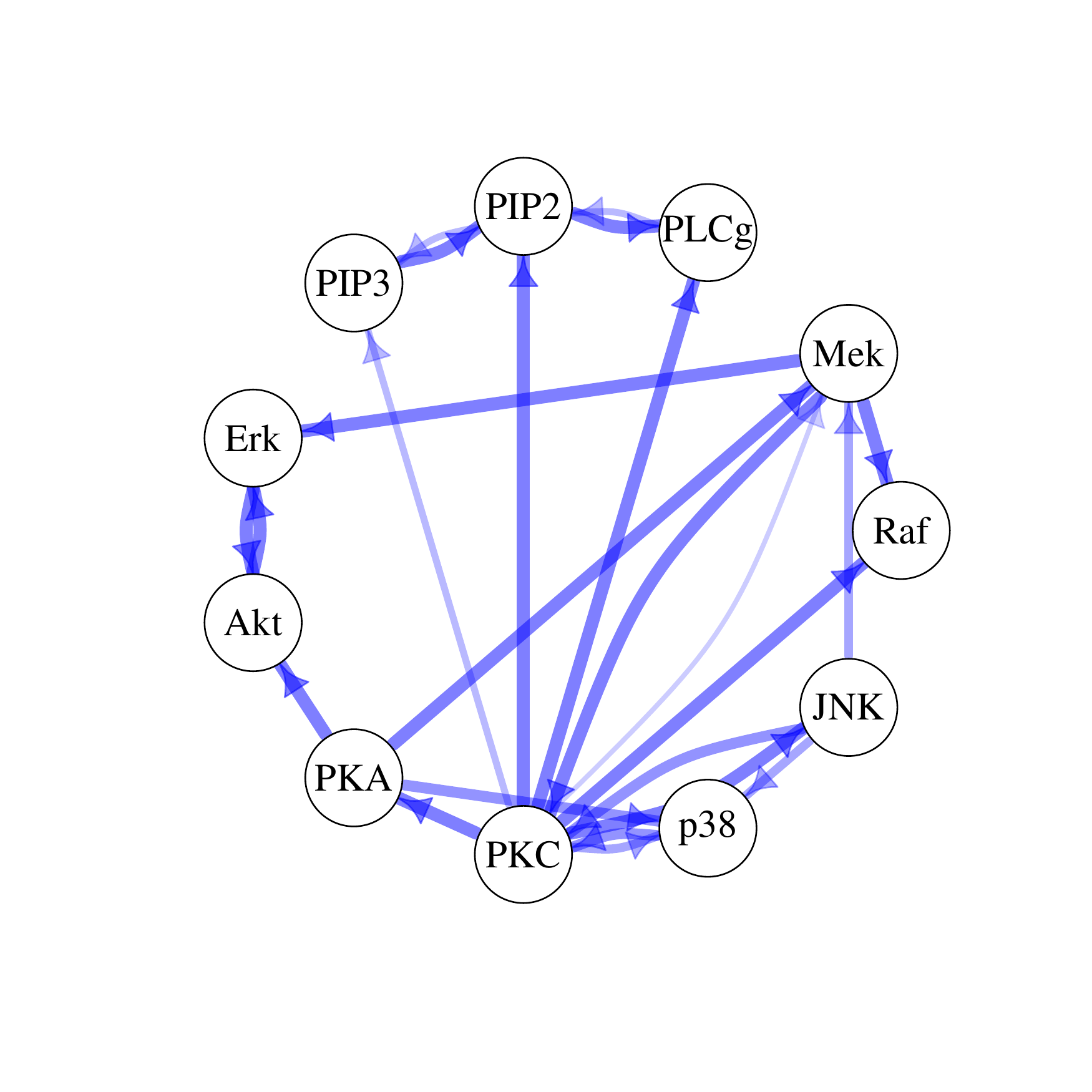}
    \label{fig:mooijcyc}
}
\hspace{-0.15cm}
\subfloat[] 
{
    \includegraphics[trim=50 80 50 70, clip, width=0.23\textwidth, keepaspectratio=true]{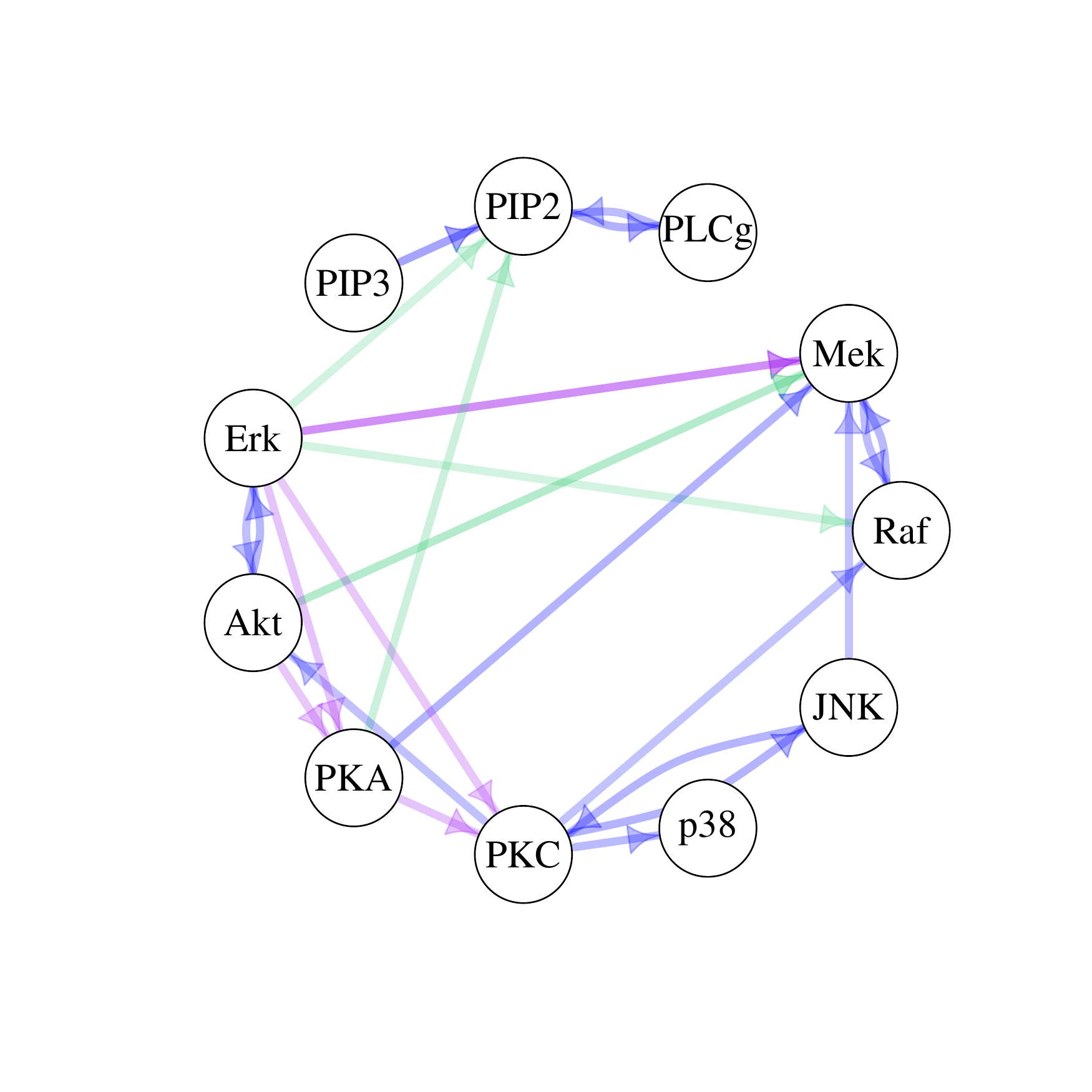}
    \label{fig:point}
}
\hspace{-0.15cm}
\subfloat[]
{
    \includegraphics[trim=50 80 50 70, clip, width=0.23\textwidth, keepaspectratio=true]{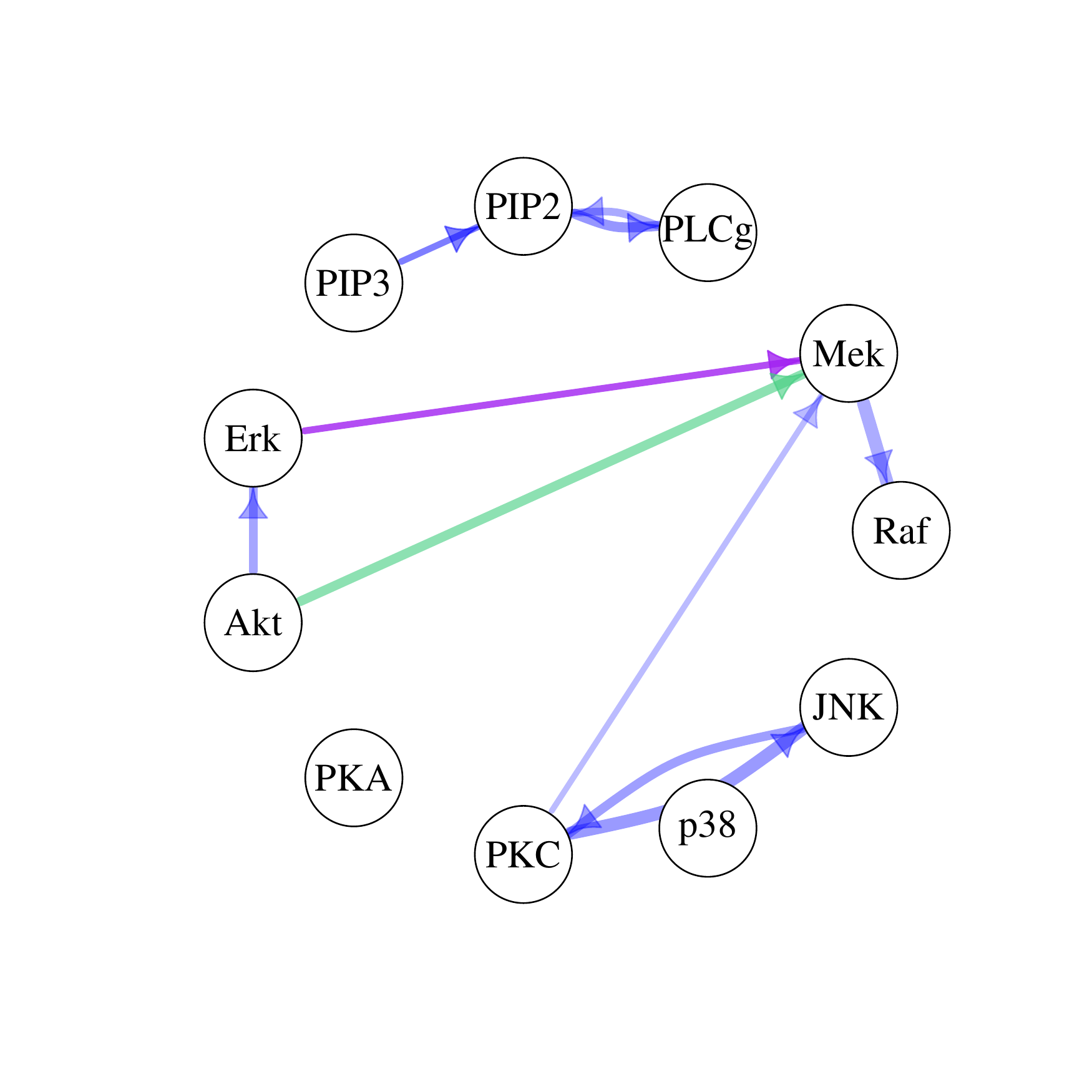}  
    \label{fig:ev5}
}
\vspace{-5pt}
\caption{{\small Flow cytometry data. 
\protect\subref{fig:cons} Union of the consensus network (according to \cite{sachs2005causal}), the reconstruction by \cite{sachs2005causal} and the best \emph{acyclic} reconstruction by \cite{MooijHeskes_UAI_13}. The edge thickness and intensity reflect in how many of these three sources that particular edge is present. 
\protect\subref{fig:mooijcyc} One of the \emph{cyclic} reconstructions by \cite{MooijHeskes_UAI_13}. The edge thickness and intensity reflect the probability of selecting that particular edge in the stability selection procedure. For more details see \cite{MooijHeskes_UAI_13}. 
\protect\subref{fig:point} \ice point estimate, thresholded at $\pm 0.35$. The edge intensity reflects the relative magnitude of the coefficients and the coloring is a comparison to the union of the graphs shown in panels~\protect\subref{fig:cons} and~\protect\subref{fig:mooijcyc}. Blue edges were also found in \cite{MooijHeskes_UAI_13} and \cite{sachs2005causal}, purple edges are reversed and green edges were not previously found in~\protect\subref{fig:cons} or~\protect\subref{fig:mooijcyc}.
\protect\subref{fig:ev5} \ice stability selection result with parameters $\mathbb{E}(V) = 2$ and $\pi_{thr} = 0.75$. The edge thickness illustrates how often an edge was selected in the stability selection procedure.} 
\label{fig:sachsdata}}
\end{centering}
\vspace{-5.5mm}
\end{figure*}

\vspace{-1.5mm}
\subsection{Flow cytometry data}\label{sec:sachs}
\vspace{-1mm}
The data published in \cite{sachs2005causal} is an instance of a data set where the external interventions differ between the environments in $\J$ and might act on several compounds simultaneously \cite{eaton2007exact}.
There are nine different experimental conditions with each  containing roughly 800 observations which correspond to measurements of the concentration of biochemical agents in single cells. The first setting corresponds to purely observational data. 

In addition to the original work by \cite{sachs2005causal}, the data set has been described and analyzed in \cite{eaton2007exact} and \cite{MooijHeskes_UAI_13}. We compare against the results of \cite{MooijHeskes_UAI_13}, \cite{sachs2005causal} and the ``well-established consensus'', according to \cite{sachs2005causal}, shown in Figures~\ref{fig:sachsdata}\subref{fig:cons} and~\ref{fig:sachsdata}\subref{fig:mooijcyc}. Figure~\ref{fig:sachsdata}\subref{fig:point} shows the (thresholded) \ice point estimate. Most of the retrieved edges were also found in at least one of the previous studies. Five edges are reversed in our estimate and three edges were not discovered previously.
Figure~\ref{fig:sachsdata}\subref{fig:ev5} shows the corresponding stability selection result with the expected number of falsely selected variables 
$\mathbb{E}(V) = 2$. This estimate is sparser in comparison to the other ones as it bounds the number of false discoveries. Notably, the feedback loops between PIP2 $ \leftrightarrow $ PLCg and PKC $\leftrightarrow$ JNK  were also found in \cite{MooijHeskes_UAI_13}. 

It is also noteworthy that we can check the model assumptions of shift interventions, which is important for these data as they can be thought of as changing the mechanism  or activity of a biochemical agent rather than regulate the biomarker directly \cite{MooijHeskes_UAI_13}. If the shift interventions are not appropriate, we are in general not able to diagonalize the differences in the covariance matrices. Large off-diagonal elements in the estimate of the r.h.s\  in~\eqref{eq:diff} indicate a mechanism change that is not just explained by a shift intervention as in~\eqref{eq:model}.  
In four of the seven interventions environments with known intervention targets the largest mechanism violation happens directly at the presumed intervention target, 
see Appendix~\ref{supp:sachs} for details. It is worth noting again that the presumed intervention target had not been used in reconstructing the network and mechanism violations.

\begin{figure*}[!tp]
\begin{centering}
\vspace{-.5cm}
\hspace{-.45cm}
\subfloat[Prices (logarithmic)]{
    \includegraphics[trim=10 10 0 0, clip, width=0.25\textwidth, keepaspectratio=true]{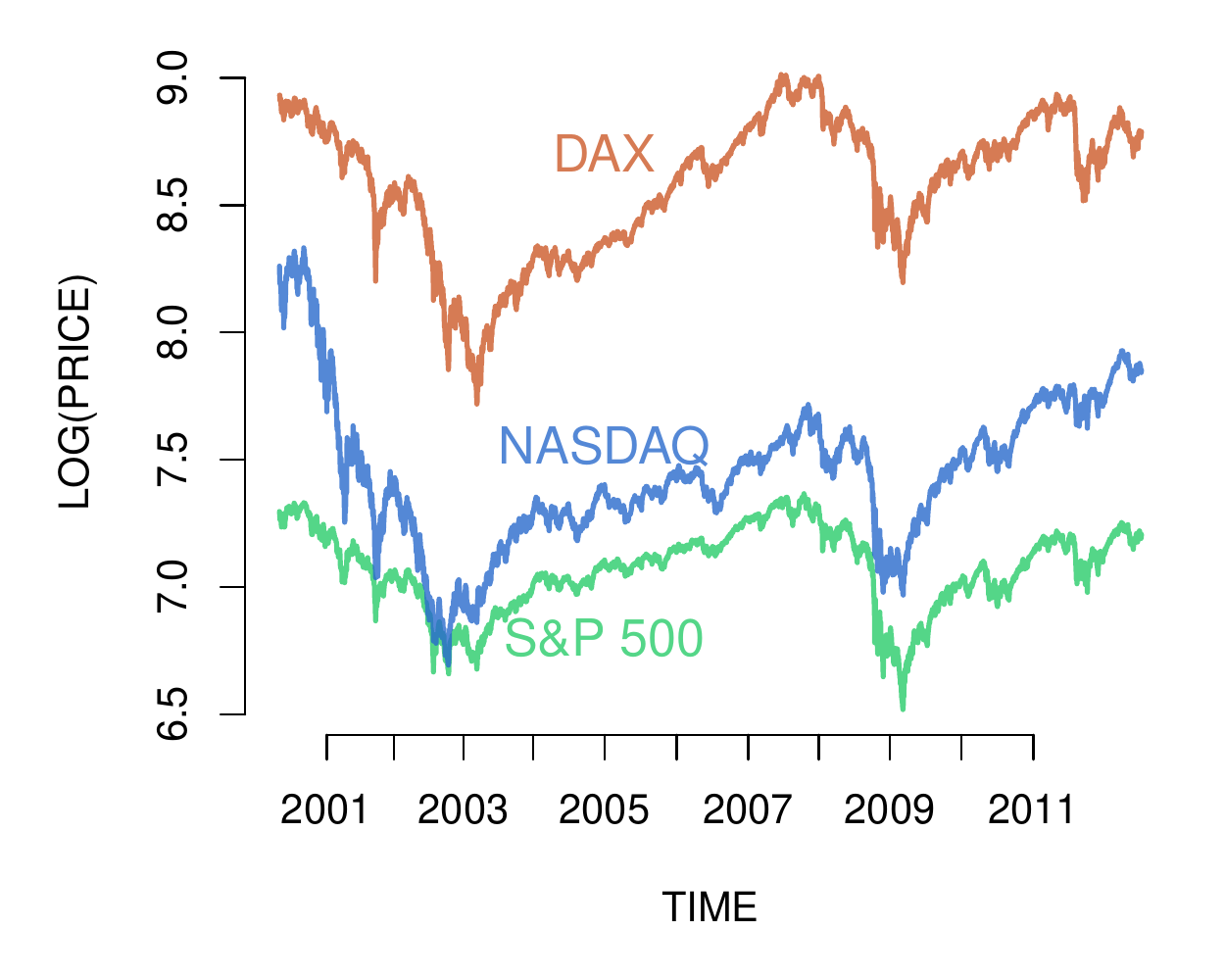}    \label{fig:1}
}
\hspace{-0.2cm}
\subfloat[Daily log-returns]{
    \includegraphics[trim=10 10 0 0, clip, width=0.25\textwidth, keepaspectratio=true]{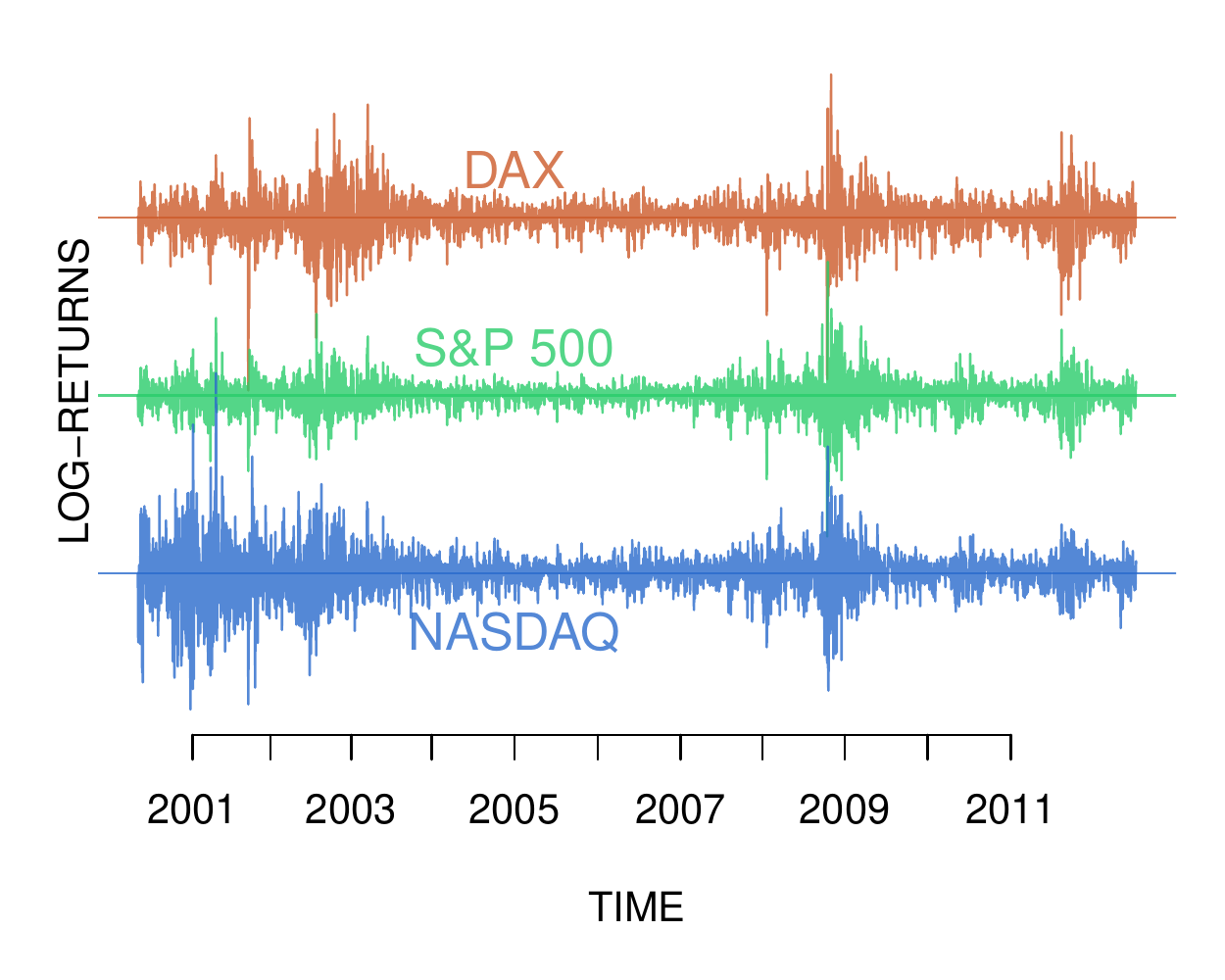}
    \label{fig:2}
}
\hspace{-0.2cm}
\subfloat[ \ice ]{
    \includegraphics[trim=10 10 0 0, clip, width=0.25\textwidth, keepaspectratio=true]{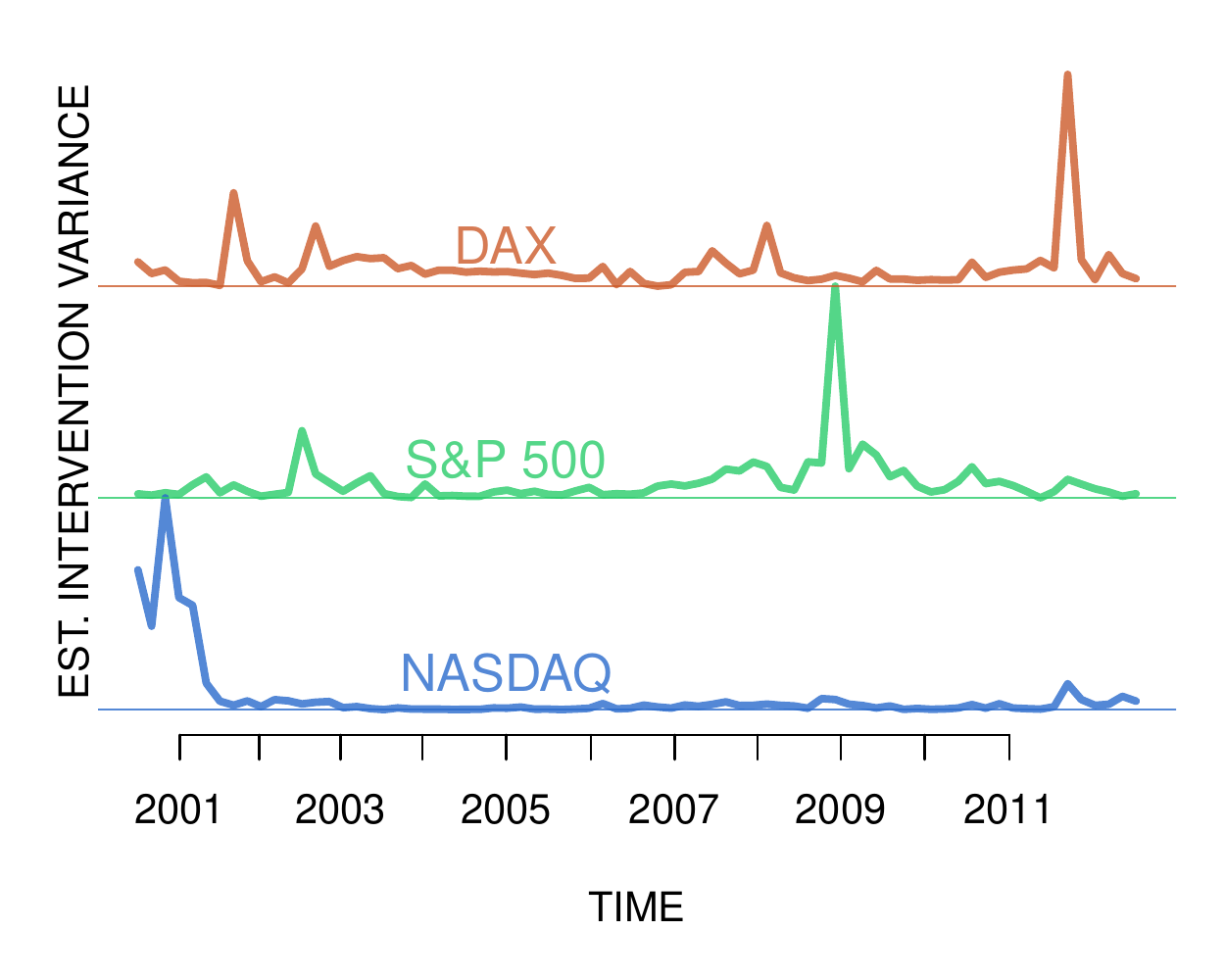}
    \label{fig:4a}
}
\hspace{-0.2cm}
\subfloat[ \ling ]{
    \includegraphics[trim=10 10 0 0, clip, width=0.25\textwidth, keepaspectratio=true]{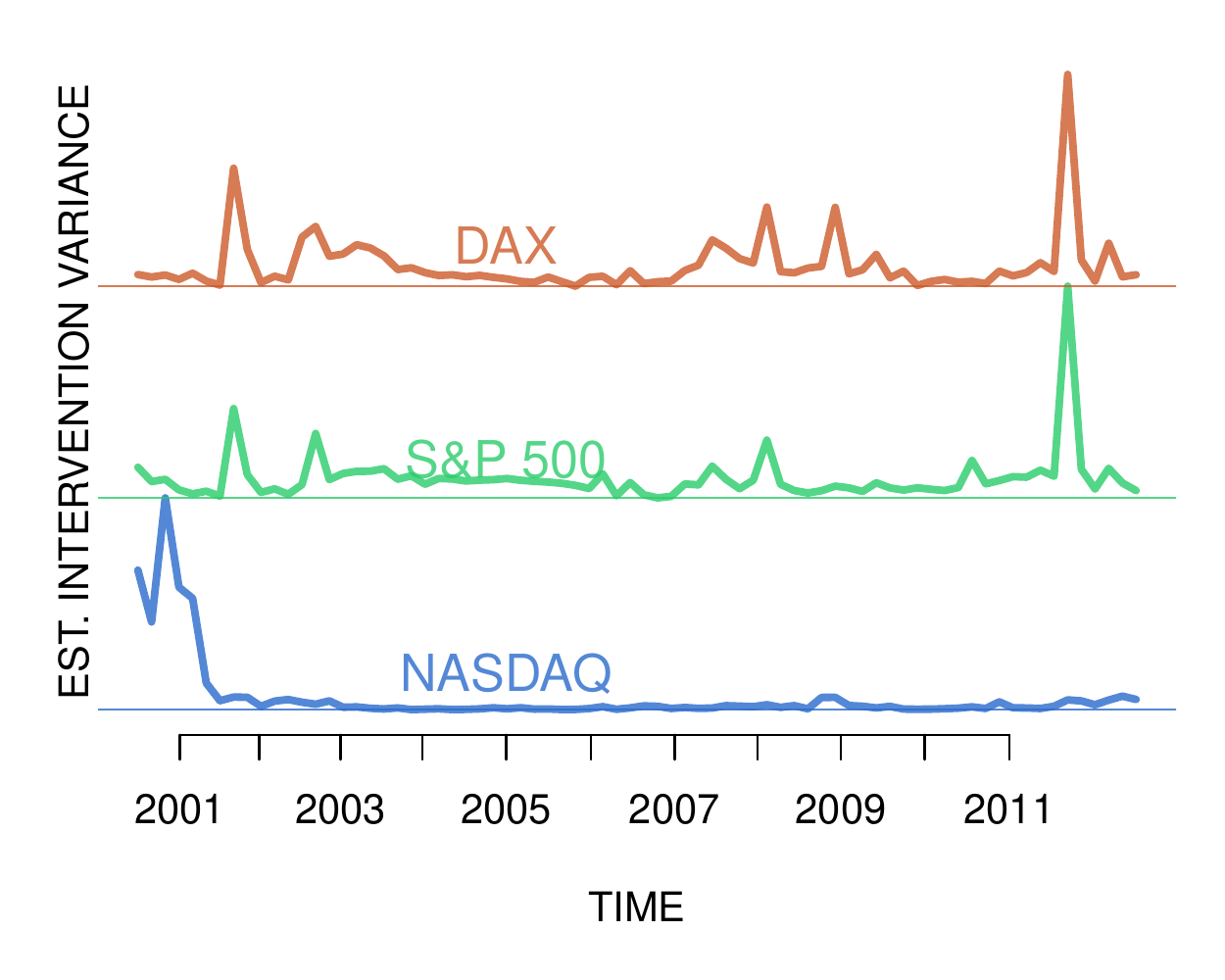}
    \label{fig:4}
}
\vspace{-5pt}
\caption{{\small  Financial time series with three stock indices:
    NASDAQ (blue; technology index), S\&P 500 (green; American equities) and DAX
    (red; German equities). (a) Prices of the three indices between May
    2000 and end of 2011 on a logarithmic scale. (b)  The scaled log-returns
    (daily change in log-price) of  the three instruments are shown.
    Three periods of increased volatility are visible starting with the dot-com
    bust on the left to the financial crisis in 2008 and the August
    2011 downturn. (c) The scaled estimated intervention variance with the
    estimated \ice network. The three down-turns are clearly separated as originating in technology, American and European equities. (d) In contrast, the analogous \ling estimated intervention variances have a peak in American equities intervention variance during the European debt crisis in 2011. }\label{fig:fin_var}}

\end{centering}
\vspace{-0.3cm}
\end{figure*}
\vspace{-0.15cm}

\subsection{Financial time series}\label{sec:finance}
Finally, we present an application in financial time series where the environment is clearly changing over time. We consider daily data from three stock indices NASDAQ, S\&P 500 and DAX for a period between 2000-2012 and group the data into 74 overlapping blocks of 61 consecutive days each. We take log-returns, as shown in panel (b) of Figure~\ref{fig:fin_var} and estimate the connectivity matrix, which is fully connected in this case and perhaps of not so much interest in itself.  It allows us, however, to estimate the intervention strength at each of the indices according to~\eqref{eq:diff_emp}, shown in panel (c). The intervention variances separate very well the origins of the three major down-turns of the markets on the period. Technology is correctly estimated by \ice to be at the epicenter of the dot-com crash in 2001 (NASDAQ as proxy), American equities during the financial crisis in 2008 (proxy is S\&P 500) and European instruments (DAX as best proxy) during the August 2011 downturn. 

\vspace{-0.1cm}
\section{Conclusion}
\vspace{-0.1cm}
We have shown that cyclic causal networks can be estimated if we  obtain covariance matrices of the variables under unknown shift interventions in different environments. 
\ice leverages solutions to the linear assignment problem and joint matrix diagonalization and the part of the computational cost that depends on the number of variables is at worst cubic. We have shown sufficient and necessary conditions under which the network is fully identifiable, which require observations from at least three different environments.  The strength and location of interventions can also be reconstructed.

\section*{References}
{\small
\bibliographystyle{unsrt}
\bibliography{bibliography}

\begin{thebibliography}{10}

\bibitem{Bollen1989}
K.A. Bollen.
\newblock {\em Structural Equations with Latent Variables}.
\newblock John Wiley \& Sons, New York, USA, 1989.

\bibitem{Spirtes2000}
P.~Spirtes, C.~Glymour, and R.~Scheines.
\newblock {\em Causation, Prediction, and Search}.
\newblock MIT Press, Cambridge, USA, 2nd edition, 2000.

\bibitem{Chickering2002}
D.M. Chickering.
\newblock Optimal structure identification with greedy search.
\newblock {\em Journal of Machine Learning Research}, 3:507--554, 2002.

\bibitem{Maathuis2009}
M.H. Maathuis, M.~Kalisch, and P.~B\"uhlmann.
\newblock Estimating high-dimensional intervention effects from observational
  data.
\newblock {\em Annals of Statistics}, 37:3133--3164, 2009.

\bibitem{Hauser2012}
A.~Hauser and P.~B{\"u}hlmann.
\newblock Characterization and greedy learning of interventional {M}arkov
  equivalence classes of directed acyclic graphs.
\newblock {\em Journal of Machine Learning Research}, 13:2409--2464, 2012.

\bibitem{Hoyer2008}
P.O. Hoyer, D.~Janzing, J.M. Mooij, J.~Peters, and B.~Sch\"olkopf.
\newblock Nonlinear causal discovery with additive noise models.
\newblock In {\em {A}dvances in {N}eural {I}nformation {P}rocessing {S}ystems
  21 ({NIPS})}, pages 689--696, 2009.

\bibitem{Shimizu2011}
S.~Shimizu, T.~Inazumi, Y.~Sogawa, A.~Hyv\"{a}rinen, Y.~Kawahara, T.~Washio,
  P.O. Hoyer, and K.~Bollen.
\newblock Direct{LiNGAM}: A direct method for learning a linear non-{G}aussian
  structural equation model.
\newblock {\em Journal of Machine Learning Research}, 12:1225--1248, 2011.

\bibitem{Mooij2011}
J.M. Mooij, D.~Janzing, T.~Heskes, and B.~Sch{\"o}lkopf.
\newblock On causal discovery with cyclic additive noise models.
\newblock In {\em {A}dvances in {N}eural {I}nformation {P}rocessing {S}ystems
  24 ({NIPS})}, pages 639--647, 2011.

\bibitem{Hyttinen2012}
A.~Hyttinen, F.~Eberhardt, and P.~O. Hoyer.
\newblock Learning linear cyclic causal models with latent variables.
\newblock {\em Journal of Machine Learning Research}, 13:3387--3439, 2012.

\bibitem{lauritzen2002chain}
S.L. Lauritzen and T.S. Richardson.
\newblock Chain graph models and their causal interpretations.
\newblock {\em Journal of the Royal Statistical Society, Series B},
  64:321--348, 2002.

\bibitem{lacerda2012discovering}
G.~Lacerda, P.~Spirtes, J.~Ramsey, and P.O. Hoyer.
\newblock Discovering cyclic causal models by independent components analysis.
\newblock In {\em Proceedings of the 24th Conference on Uncertainty in
  Artificial Intelligence (UAI)}, pages 366--374, 2008.

\bibitem{scheines2010combining}
R.~Scheines, F.~Eberhardt, and P.O. Hoyer.
\newblock Combining experiments to discover linear cyclic models with latent
  variables.
\newblock In {\em International Conference on Artificial Intelligence and
  Statistics ({AISTATS})}, pages 185--192, 2010.

\bibitem{Pearl2009}
J.~Pearl.
\newblock {\em Causality: Models, Reasoning, and Inference}.
\newblock Cambridge University Press, New York, USA, 2nd edition, 2009.

\bibitem{Eberhardt2010}
F.~Eberhardt, P.~O. Hoyer, and R.~Scheines.
\newblock Combining experiments to discover linear cyclic models with latent
  variables.
\newblock In {\em International Conference on Artificial Intelligence and
  Statistics ({AISTATS})}, pages 185--192, 2010.

\bibitem{peters2015causal}
J.~Peters, P.~B{\"u}hlmann, and N.~Meinshausen.
\newblock Causal inference using invariant prediction: identification and
  confidence intervals.
\newblock {\em Journal of the Royal Statistical Society, Series B, to appear.},
  2015.

\bibitem{Jackson2003}
A.L. Jackson, S.R. Bartz, J.~Schelter, S.V. Kobayashi, J.~Burchard, M.~Mao,
  B.~Li, G.~Cavet, and P.S. Linsley.
\newblock Expression profiling reveals off-target gene regulation by {RNAi}.
\newblock {\em Nature Biotechnology}, 21:635--637, 2003.

\bibitem{kulkarni2006evidence}
M.M. Kulkarni, M.~Booker, S.J. Silver, A.~Friedman, P.~Hong, N.~Perrimon, and
  B.~Mathey-Prevot.
\newblock Evidence of off-target effects associated with long dsrnas in
  drosophila melanogaster cell-based assays.
\newblock {\em Nature methods}, 3:833--838, 2006.

\bibitem{eaton2007exact}
D.~Eaton and K.~Murphy.
\newblock Exact {B}ayesian structure learning from uncertain interventions.
\newblock In {\em International Conference on Artificial Intelligence and
  Statistics ({AISTATS})}, pages 107--114, 2007.

\bibitem{Eberhardt2007}
F.~Eberhardt and R.~Scheines.
\newblock Interventions and causal inference.
\newblock {\em Philosophy of Science}, 74:981--995, 2007.

\bibitem{Korb2004}
K.~Korb, L.~Hope, A.~Nicholson, and K.~Axnick.
\newblock Varieties of causal intervention.
\newblock In {\em Proceedings of the Pacific Rim Conference on AI}, pages
  322--331, 2004.

\bibitem{Tian2001}
J.~Tian and J.~Pearl.
\newblock Causal discovery from changes.
\newblock In {\em Proceedings of the 17th Conference Annual Conference on
  Uncertainty in Artificial Intelligence ({UAI})}, pages 512--522, 2001.

\bibitem{sachs2005causal}
K.~Sachs, O.~Perez, D.~Pe'er, D.~Lauffenburger, and G.~Nolan.
\newblock Causal protein-signaling networks derived from multiparameter
  single-cell data.
\newblock {\em Science}, 308:523--529, 2005.

\bibitem{Ziehe04afast}
A.~Ziehe, P.~Laskov, G.~Nolte, and K.-R. M\"uller.
\newblock A fast algorithm for joint diagonalization with non-orthogonal
  transformations and its application to blind source separation.
\newblock {\em Journal of Machine Learning Research}, 5:801--818, 2004.

\bibitem{lasp2013}
R.E. Burkard.
\newblock Quadratic assignment problems.
\newblock In P.~M. Pardalos, D.-Z. Du, and R.~L. Graham, editors, {\em Handbook
  of Combinatorial Optimization}, pages 2741--2814. Springer New York, 2nd
  edition, 2013.

\bibitem{meinshausen2010stability}
N.~Meinshausen and P.~B{\"u}hlmann.
\newblock Stability selection.
\newblock {\em Journal of the Royal Statistical Society, Series B},
  72:417--473, 2010.

\bibitem{MooijHeskes_UAI_13}
J.M. Mooij and T.~Heskes.
\newblock Cyclic causal discovery from continuous equilibrium data.
\newblock In {\em Proceedings of the 29th Annual Conference on {U}ncertainty in
  {A}rtificial {I}ntelligence ({UAI})}, pages 431--439, 2013.

\end{thebibliography}
}

\newpage
\appendix
\begin{center}
Appendix to \vspace{0.4cm}\\
{\Large \bf \ices: Learning causal cyclic graphs from unknown shift interventions}
\end{center}
\section{Identifiability -- Proof of Theorem \ref{theorem:1}}\label{sec:supp_ident}

\begin{proof}
``if'':
Let ${\D'}$ be a solution of \eqref{eq:Dhatpop}. Let us write $\D_{m \bullet}'$ for the $m$-th row of $ \D'$ and $\D_{m \bullet}$ for the $m$-th row of ${\D}$, $m=1,\ldots,p$. Furthermore let us define $\mb g_m := {\D}^{-T} \D_{m \bullet}'$, $m=1,\ldots,p$. 
We will show that at most one entry of this vector is nonzero. Note that by equation~\eqref{eq:diff} we have $   \DSxj=  {\D}^{-1} \DScj{\D}^{-T}$ for all $j \in \J$. By equation~\eqref{eq:diff}, $L(\D \DSxj \D^T)=0$. As $\D'$ solves equation~\eqref{eq:Dhatpop}, this implies  $L( \D' \DSxj \D'^T ) =0$ for all $j \in \J$. Hence the offdiagonal elements of $\D' \DSxj \D'^T$ are zero, which implies
\begin{align*}
	\mb g_{m'} \perp \DScj \mb g_{m} \mbox{ for all } m' \neq m \mbox{ and for all $j \in \J$}. 
\end{align*}
As the $\mb g_{m'}$ are linearly independent, this implies that for all pairs $j,j'\in \J$, $ \DScj \mb g_m$ and $ \DScjp \mb g_m$ 
are collinear i.e.\ for all $(j,j')$ there exists a $\lambda_{j,j'} \in \mathbb{R}$ such that $ \DScj \mb g_m = \lambda_{j,j'} \DScjp \mb g_m$ or $ \lambda_{j,j'} \DScj \mb g_m = \DScjp \mb g_m$ 


Take arbitrary $k,l \in \{ 1, \ldots, p\}$ and choose $j,j' \in \J$ such that \eqref{eq:manyunique} is satisfied. By the argumentation above, there exists a $\lambda_{j,j'} \in \mathbb{R}$ such that $ \DScj \mb g_m = \lambda_{j,j'} \DScjp \mb g_m$ or $ \lambda_{j,j'} \DScj \mb g_m = \DScjp \mb g_m$. Without loss of generality let us assume the latter. Recall that both $\DScj$ and $\DScjp$ are diagonal matrices.
Now condition \eqref{eq:manyunique} implies that the $k$-th or the $l$-th entry on the diagonal of $\lambda_{j,j'} \ \DScj-  \DScjp$ is nonzero (or both). Hence, the $k$-th or the $l$-th entry of $ \mb g_m$ s zero (or both). By repeating this argumentation for all $k$ and $l$, at most one entry of $ \mb g_m$ is nonzero. Thus, $\D_{m \bullet}' = \D^T \mb g_m = ( \mb g_m^T \D)^T$ is a multiple of one of the rows of ${\D}$. 

By applying this argumentation for all $m=1,\dots,p$, each row of $\D'$ is a multiple of one of the rows of $\D$. As both ${\D}$ and ${\D'}$ are invertible, there exists a bijection between the rows of ${\D'}$ and ${\D}$ such that the corresponding rows are collinear. Furthermore, the diagonal of $ \D'$ and $\D$ is $(1,\ldots,1)$. Hence let us consider a bijection $\sigma : \{1,\ldots,p\} \mapsto \{1,\ldots,p \}$ such that the $\sigma(m)$-th row of $\D'$ is a multiple of the $m$-th row of $\D$, i.e. $\frac{1}{\D_{\sigma(m), m}'}   \D_{\sigma(m) \bullet}'= \D_{m \bullet}$ for all $m=1, \ldots,p$. 
We want to show that this bijection is the identity. First observe that, as the diagonal of $\D'$ and $\D$ is $(1,\ldots,1)$, $ \frac{1}{\D_{\sigma(m),m}'} = \D_{m,\sigma(m)}$ for all $m=1,\ldots,p$. Now let us consider a cycle in this permutation
, i.e. $m_1,\ldots,m_{\eta+1}=m_1$, $\eta > 1$, $m_\iota \neq m_{\kappa}$ for $ 1 \le \iota < \kappa \le \eta$ and with $\sigma(m_\iota)=m_{\iota+1}$ for $1 \le \iota \le \eta$. 
If this leads to a contradiction, we can conclude that $\sigma$ is the identity.
As $\D_{m,m} =1 $, $ \D_{\sigma(m),m}' \neq 0$, i.e. $ \D_{m_{\iota+1},m_{\iota}}'\neq 0$ for $ 1\le \iota \le \eta$. This corresponds to a cycle in with product 
\begin{equation}\label{eq:cycprod}
	\prod_{\iota =1,\ldots,\eta} \D_{m_{\iota+1},m_\iota}' = \prod_{\iota=1,\ldots,\eta} \frac{1}{\D_{m_{\iota},m_{\iota+1}} }.
\end{equation}
As $\D'$ is a solution of \eqref{eq:Dhatpop}, $\CP(\mb I - \D')<1$, hence the product on the left hand side of equation \eqref{eq:cycprod} is in absolute value strictly smaller than $1$, see~\eqref{CP}. Analogously, as $ \D_{m_{\iota},m_{\iota+1}} \neq 0$ for $\iota =1,\ldots,\eta$, the sequence $m_{\eta+1},m_{\eta},\ldots,m_{1}$ 
corresponds to a cycle with product
\begin{equation*}
	\prod_{\iota=1,\ldots,\eta} \D_{m_{\iota},m_{\iota+1}}. 
\end{equation*}
Using the same argumentation as for $\D'$, this product is in absolute value strictly smaller than $1$, which contradicts~\eqref{eq:cycprod}.
Hence such cycles of length $\geq 2$  do not exist and $\sigma$ is the identity. Hence, $\D' = \D$.

``only if'':
As above define $\D_{m \bullet}$ as the $m$-th row of $\D$ and let us write $\mb u_m \in \mathbb{R}^p$ for the $m$-th unit vector for $m=1,\ldots,p$. Assume that \eqref{eq:manyunique} is not true, i.e. there exist $k,l \in \left\{ 1,\dots,p \right\}$ such that for all $j,j' \in \J$,
\begin{equation}\label{eq:etaequality}
	(\DScj)_{kk} ( \DScjp )_{ll} = (\DScj)_{ll} ( \DScjp)_{kk}. 
\end{equation}
Without loss of generality let us fix a  $j' \in \J$ with $ ( \DScjp )_{kk} \neq 0$
, and define $ \lambda :=  ( \DScjp )_{ll} / ( \DScjp )_{kk}$. If such a $j'$ does not exist, we can apply the same argumentation as below but with the $k$ and $l$ interchanged and $\lambda := 0$. 

Note that the definition of $\lambda$ does not depend on $j$ and that by equation~\eqref{eq:diff} we have $   \DSxj=  {\D}^{-1} \DScj{\D}^{-T}$. Then, for $ \delta \in \mathbb{R}$ we can define $\D_{k \bullet}' := \D_{k \bullet} + \delta \D_{l \bullet}$ and $ \D_{l \bullet}' := \D_{l \bullet} - \delta \lambda \D_{k \bullet}$ and we obtain for all $j \in \J$ 
\begin{align*}
	\D_{l \bullet}'^T \DSxj \D_{k \bullet}' &= (\mb u_l - \delta \lambda \mb u_k)^T \DScj (\mb u_k+\delta \mb u_l) \\
	&= \delta (\DScj)_{ll}-\delta \lambda (\DScj)_{kk} \\
	&= 0.
\end{align*}
In the second equation we used \eqref{eq:etaequality}. Furthermore, for small $\delta$ let us scale $\D_{k \bullet}'$ such that the $k$-th component of the vector is $1$. Analogously, let us scale $\D_{l \bullet}'$ such that the $l$-th component of the vector is $1$. Then we can define the matrix ${\D'}$ as the rows of ${\D}$ except for row $k$ and $l$ which are replaced by $\D_{k \bullet}'$ and $\D_{l \bullet}'$. By above reasoning, this matrix satisfies 
\begin{equation*}
	{\D'} \DSxj {\D'}^T \in \mbox{Diag}(p) 
\end{equation*}
for all $j \in \J$ and ${\D'}$ is invertible. Furthermore, the diagonal elements of ${\D'}$ are $1$. Recall that the path-products of $\mb I - \D $ over cycles are in absolute value smaller than $1$, see \eqref{CP}. For small $\delta$, $ \mb I - {\D'} $ is close to $ \mb I - {\D} $ (in an arbitrary matrix norm) and hence the path products of $\mb I - \D' $ over cycles are in absolute value smaller than $1$ as well. As $\D$ is invertible, $\D' \neq \D$. Hence the solution to \eqref{eq:Dhatpop} is not unique. This concludes the proof.

\end{proof}
\section{Polynomial-time algorithm}\label{sec:supp_algo}

Here, we provide the necessary theoretical result to show that \ice has a computational cost of $O(\vert \J \vert \cdot n \cdot p^2)$. Specifically, we show that Step~3 in Algorithm~\ref{alg:hiddenice} can be cast in terms of the classical linear sum assignment problem, having a computational complexity of $O(p^3)$.

\begin{theorem}\label{theorem:algo}
	Let $\D \in \mathbb{R}^{p \times p}$ be a matrix with $\CP(\D)<1$, $\mbox{diag}(\D) \equiv 1$ and $ \D_{k,l} \neq 0$ for $k,l \in \{1, \ldots,p \}$. For $\D' \in \mathbb{R}^{p \times p}$ define
\begin{equation*}
	P(\D') := \prod_{k,l} |\D_{k,l}' |.
\end{equation*}
Furthermore define 
\begin{align*}
	\mathcal{D}_p := \{ \D' : &\text{ There exists a permutation $\sigma$ of $\{ 1,\ldots,p \}$ such that the $\sigma(m)$-th row of $\D$ } \\ &\text{ is collinear to 
  the $m$-th row of $\D'$ and diag$(\D') \equiv 1$ } \}.
\end{align*}
Then,
\begin{equation*}
	\D = \arg \min_{\D' \in \mathcal{D}_p} P(\D') = \arg \min_{\D'
          \in \mathcal{D}_p} \log P(\D').
\end{equation*}
\end{theorem}

\begin{proof}
	Let $\D' \in \mathcal{D}_p$ with $\D' \neq \D$. Let us write $\D_{m \bullet}$ for the $m$-th row of $\D$ and analogously $\D_{m \bullet}'$ for the $m$-th row of $\D'$, $m=1,\ldots,p$. Now let $\sigma$ be a permutation such that the $\sigma(m)$-th row of $\D$ is collinear to the $m$-th row of $\D'$. As $\D' \neq \D$, we have that $\sigma \neq \mbox{Id}$. As diag$(\D') \equiv 1$,
\begin{equation*}
	\frac{1}{\D_{\sigma(m),m}} \D_{\sigma(m)\bullet} = \D_{m \bullet}'.	
\end{equation*}
It immediately follows that
\begin{equation*}
	  \left( \prod_{m=1,\ldots,p}  \frac{1}{|\D_{\sigma(m),m}|} \right)^p  P(\D) = P(\D').
\end{equation*}
As $\CP(\D)<1$ and $\sigma$ is not the identity, $\prod_{m=1,\ldots,p} |\D_{\sigma(m),m} |< 1$. As all elements of $\D$ and $\D'$ are nonzero, $P(\D)>0$ and $P(\D')>0$. Hence, $P(\D') > P(\D)$. This concludes the proof.

\end{proof}

\textbf{Remark:} We can define the relative loss function of moving row $k$ to row $l$ as 
\begin{equation*}
	\ell(k,l) = -\log(|\D_{k,l}'|) .
\end{equation*}
Then the linear assignment problem that minimizes this problem also yields the correct permutation for Step~3 in Algorithm~\ref{alg:hiddenice} if it exists, i.e. the permutation $\sigma$ on $\{1, \ldots, p \}$ that minimizes
\begin{equation*}
	\sum_{k=1}^p \ell(k,\sigma(k))
\end{equation*}
satisfies that $\D_{m \bullet}'$ is collinear to  $ \D_{\sigma(m) \bullet}$.\\

\textbf{Remark:} Allowing for self-loops would lead to an identifiability problem, independent of the method. For every model with self-loops and $\CP<1$ there is a model without self-loops and $\CP \le 1$ yielding the same observational distribution in equilibrium. The connectivity matrix without self-loops can thus be seen as a representative of a whole class of connectivity matrices that allow self-loops. 
Specifically, if the connectivity matrix with self-loops is $\B^*$, define matrix ${\mb T}$ by $\texttt{PermuteAndScale}({\mb I}-\B^*)={\mb T}({\mb I}-\B^*)$, where $\texttt{PermuteAndScale}()$ is the operation defined in Step 3 of the \ice algorithm. Technically, $\texttt{PermuteAndScale}()$ is only defined for matrices that are nonzero outside of the diagonal. Using similar arguments as in Theorem~\ref{theorem:algo}, $\texttt{PermuteAndScale}()$ can be extended to arbitrary matrices with nonzero diagonal elements. To be more precise, there exists a matrix $\mb T$ such that $\CP({\mb T}({\mb I}-\B^*))\le 1$, $\mbox{diag}({\mb T}({\mb I}-\B^*)) \equiv 1$ and such that $\mb T$ is the product of a diagonal scaling matrix with a permutation matrix. 
Then define $\B_{new}:={\mb I}-{\mb T}({\mb I}-\B^*)$, ${\mb e}_{j,new} = {\mb T} {\mb e}_{j}$ and ${\mb c}_{j,new} = {\mb T} {\mb c_j}$ for all $j \in \J$. As ${\mb T}$ is the product of a diagonal scaling matrix with a permutation matrix, assumptions (B) and (C) are still fulfilled and 
$\x_{j,new} = ({\mb I}-\B_{new})^{-1}({\mb e}_{j,new}+{\mb c}_{j,new}) = ({\mb I}-\B^*)^{-1} ({\mb e}_{j}+{\mb c}_{j}) =\x_{j}$ for all $j \in \J$. This implies that the two matrices $\B^*$ with self-loops and $\B_{new}$ without self-loops (since it has zeroes on the diagonal by construction) have both $\CP \le 1$ and yield the same distribution.


\section{Intervention variances and model misspecification}\label{supp:sachs}
The method allows to validate and check the assumptions to some
extent. This is especially important in the data of
\cite{sachs2005causal} as pointed out in
\cite{MooijHeskes_UAI_13}. The interventions can mostly be thought
of as not changing the concentration of a biochemical agent but rather
changing the activity of the agent, for example by inhibiting the
reactions in which the agent is involved \cite{MooijHeskes_UAI_13}. Under such a mechanism
change, it is doubtful whether the interventions are well approximated
by our model~\eqref{eq:modelinterv} with independent
shift-interventions. We can check the assumptions by the success of
the joint diagonalization procedure. Specifically, we get an
empirical version of~\eqref{eq:diff} when plugging in the estimators and can check whether all
off-diagonal elements on the right hand side of~\eqref{eq:diff} are
small or vanishing. We list below results for the seven experimental
intervention conditions whose target is well described in
\cite{MooijHeskes_UAI_13}.  
The element on the right-hand side of~\eqref{eq:diff} with
the largest absolute value is selected. We use now the Gram instead of
the covariance matrix to be also sensitive to model-violations of the
additional assumption (C'), see Section~\ref{sec:Gram},
 though the results are almost identical
whether using the Gram or covariance matrix.  These large off-diagonal elements indicate a violated mechanism in the
sense that the model~\eqref{eq:modelinterv} does not fit very well,
because either  the interventions have not been of the assumed
shift-type or the causal mechanism in which the agent is involved has
changed under the intervention.  
\begin{center}
\begin{tabular}{llll}
Experiment & Reagent & Intervention & largest mechanism violation \\ \hline
3 &Akt-Inhibitor&inhibits AKT activity& PLCg $\leftrightarrow$ PKA\\
4 &G0076&inhibits {\bf PKC} activity& {\bf PKC} $\leftrightarrow$ PIP2
\\
5 &Psitectorigenin&inhibits {\bf PIP2} abundance& {\bf PIP2} $\leftrightarrow$ PKA\\
6 &U0126&inhibits {\bf MEK} activity& {\bf MEK} $\leftrightarrow$ PKA \\
7 &LY294002&changes PIP2/PIP3 mechanisms& PKA $\leftrightarrow$ JNK\\
8 &PMA&activates PKC activity& MEK $\leftrightarrow$ PKA\\
9 &$\beta$2CAMP&activates {\bf PKA} activity& {\bf PKA} $\leftrightarrow$ PKC\\ \hline
\end{tabular}
\end{center}
The table above lists the results for the seven experimental conditions
where we know the intervention mechanism, at least approximately. 
The results are interesting in that the most violated mechanism (the
largest entry in the off-diagonal matrix on the right-hand side of the empirical
version of~\eqref{eq:diff}) occurs
in 4 of the 7 experimental conditions directly at the intervention
target. In 3 of these 4 cases, the violated mechanism concerns a
relation that has a large entry in the estimated connectivity matrix.  This corresponds well with the model of activity
interventions in~\cite{MooijHeskes_UAI_13}. Note that we have not made use of the
intervention targets in the estimation procedure.  The interesting point is
that  we can use the model violations to estimate with some success where the
interventions occurred.

\section{Beyond covariances}\label{sec:Gram}
For the method above, we exploit differences in the covariance of
observations across different environments. We can also exploit a
shift in the mean of the intervention strength $\mb c$ (and consequently
in the observations $\x$) when strengthening the condition
(C) to (C'). Specifically, we  require for (C') that in each environment $j\in \J$ 
the shift in the mean $E(\mb c_j)$ equals zero for all variables
except at most one variable. The variable with a non-zero shift in the mean can change
from one environment to another. Note that the counterpart
of~\eqref{eq:Si}  when using the Gram matrix instead of the covariance
matrix reads
\begin{align}
(\mb I - \B) \G_{\x,j} (\mb I - \B)^T &= \Gcj + \Ge \label{eq:Gi} .
\end{align}
Under the stronger version (C'), the difference across
environments  of the
right-hand side in~\eqref{eq:Gi}  is again a diagonal matrix and we
can proceed just as above, by replacing the covariance matrices with
Gram matrices throughout. If the assumption (C') is satisfied, this
allows identifiability of the graph in a wider range of settings
(Theorem~\ref{theorem:1} can be adapted in a straightforward manner by
again replacing covariances with Gram matrices) but requires the
stricter condition (C'). Since in practice it is often unclear whether
the stricter condition is approximately true, we work mainly with the
weaker assumption (C) and exploit only shifts in the covariance matrices. 



\section{Additional figures} \label{sec:appFig5}
\begin{table}[ht]
\hspace{-0.0cm}
  \begin{tabular}
      {llllll} 
	& 
	 \hspace{1mm}     
     {\small \emph{Setting 1}} & 
     \hspace{-2mm}
     {\small \emph{Setting 2}} & 
     \hspace{-2mm}
     {\small \emph{Setting 3}} & 
     \hspace{-2mm}
     {\small \emph{Setting 4}} & 
     \hspace{-2mm}
     {\small \emph{Setting 5}} \\      
      
     & 
	 \hspace{1mm}     
     {\small $n = 1000$} & 
     \hspace{-2mm}
     {\small $n = 10000$} & 
     \hspace{-2mm}
     {\small $n = 10000$} & 
     \hspace{-2mm}
     {\small $n = 10000$} & 
     \hspace{-2mm}
     {\small $n = 10000$} \\

     & 
	 \hspace{1mm}
     {\small no hidden vars.} & 
     \hspace{-2mm}
     {\small no hidden vars.} & 
     \hspace{-2mm}
     {\small hidden vars.} & 
     \hspace{-2mm}
     {\small no hidden vars.} & 
     \hspace{-2mm}
     {\small no hidden vars.} \\
 
 	& 
 	\hspace{1mm}
 	{\small $\iMult = 1$} & 
 	 \hspace{-2mm}
     {\small $\iMult = 1$} & 
     \hspace{-2mm}
     {\small $\iMult = 1$} & 
     \hspace{-2mm}
     {\small $\iMult = 0$} & 
     \hspace{-2mm}
     {\small $\iMult = 0.5$} 
     \vspace{-0.65cm}\\

\rotatebox{90}{\ice} & 
\hspace{-3mm}
\subfloat{
  \raisebox{-.1\height}{\includegraphics[trim=85 85 50 02, clip, width=0.175\textwidth, keepaspectratio=true]{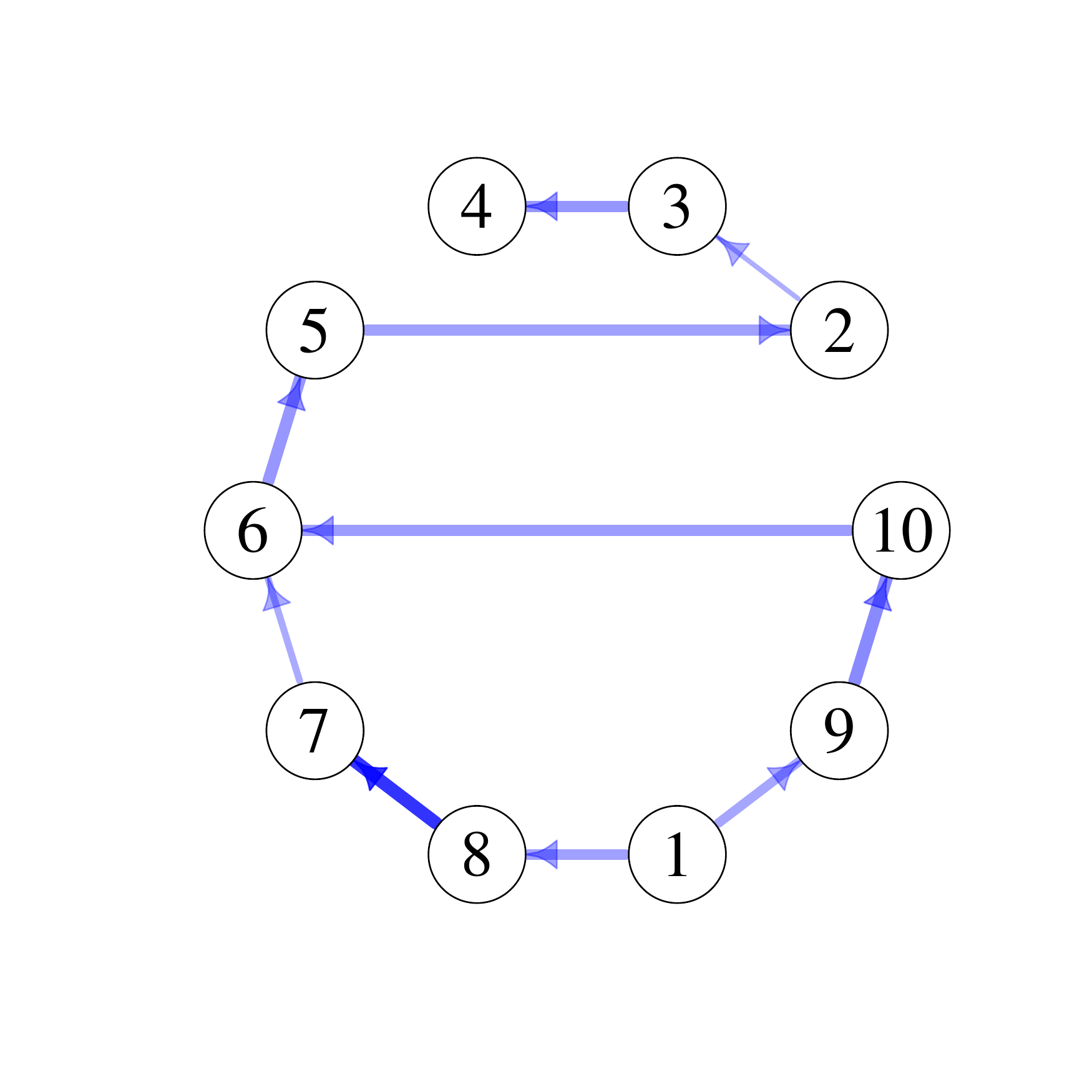}
}} & \hspace{-6mm}
\subfloat{
     \raisebox{-.1\height}{\includegraphics[trim=85 85 50 70, clip, width=0.175\textwidth, keepaspectratio=true]{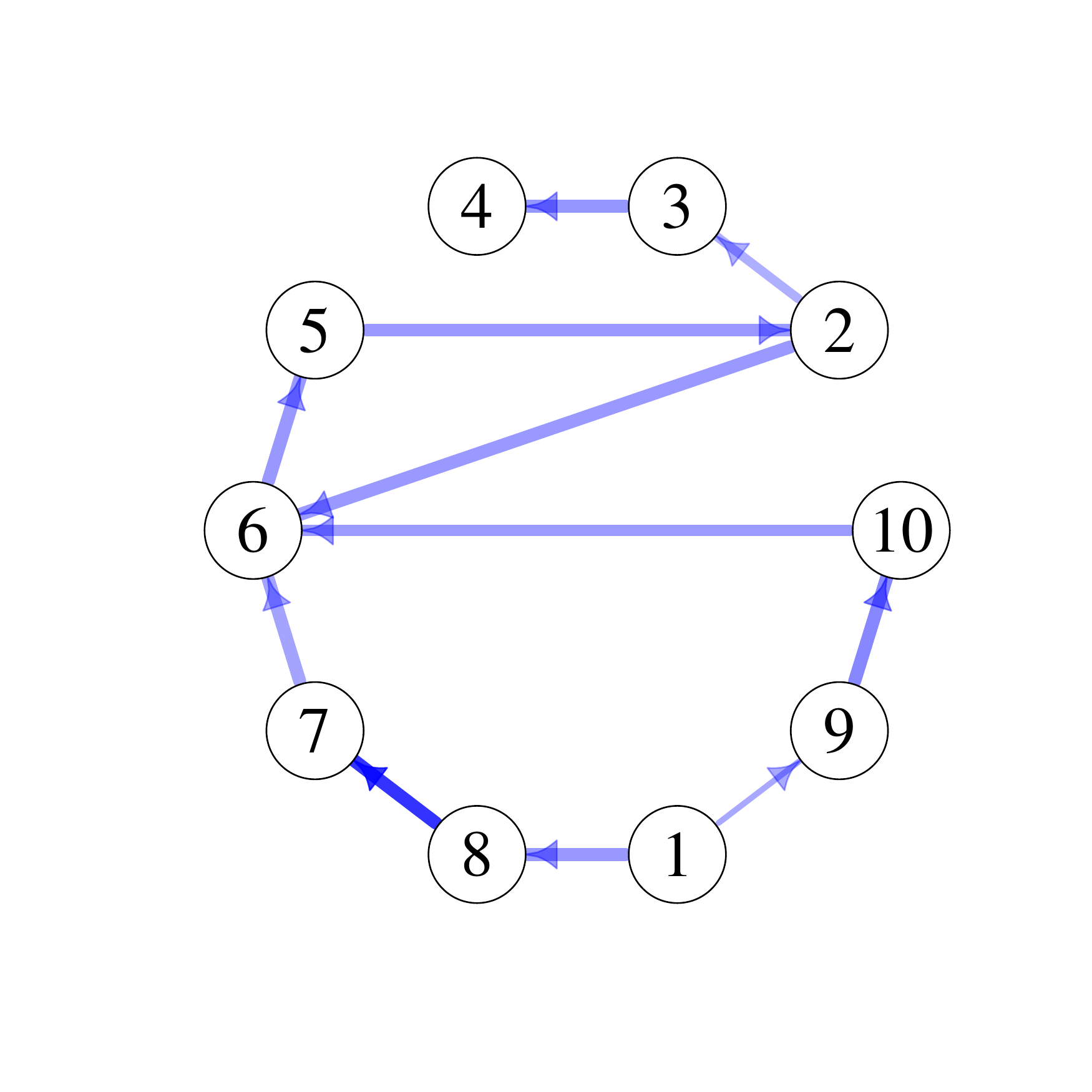}
}} & \hspace{-6mm}
\subfloat{
     \raisebox{-.1\height}{\includegraphics[trim=85 85 50 70, clip, width=0.175\textwidth, keepaspectratio=true]{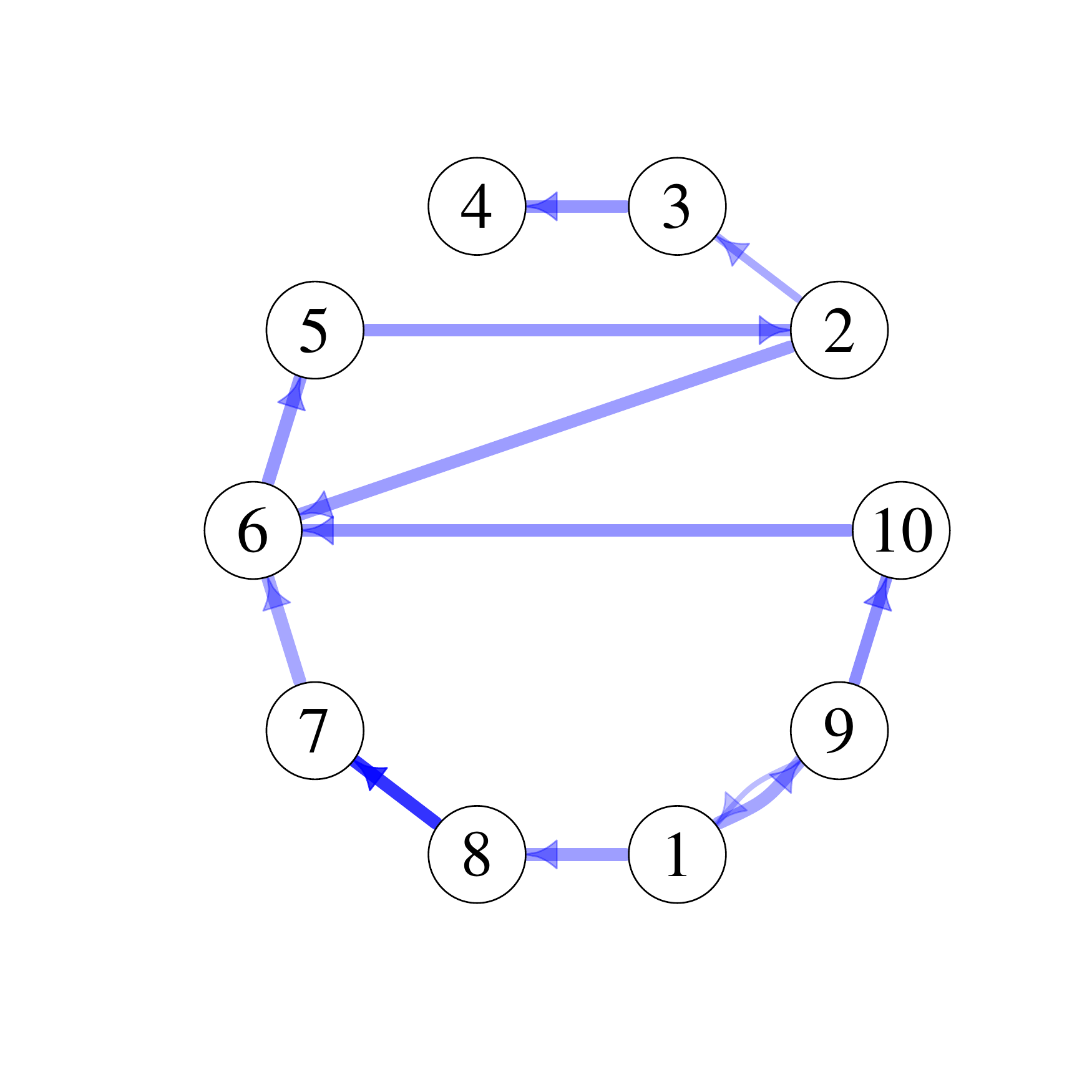}
}} & \hspace{-6mm}
\subfloat{
     \raisebox{-.1\height}{\includegraphics[trim=85 85 50 70, clip, width=0.175\textwidth, keepaspectratio=true]{figs/synthetic/new/im0/empty}
}} & \hspace{-6mm}
\subfloat{
    \raisebox{-.1\height}{\includegraphics[trim=85 85 50 70, clip, width=0.175\textwidth, keepaspectratio=true]{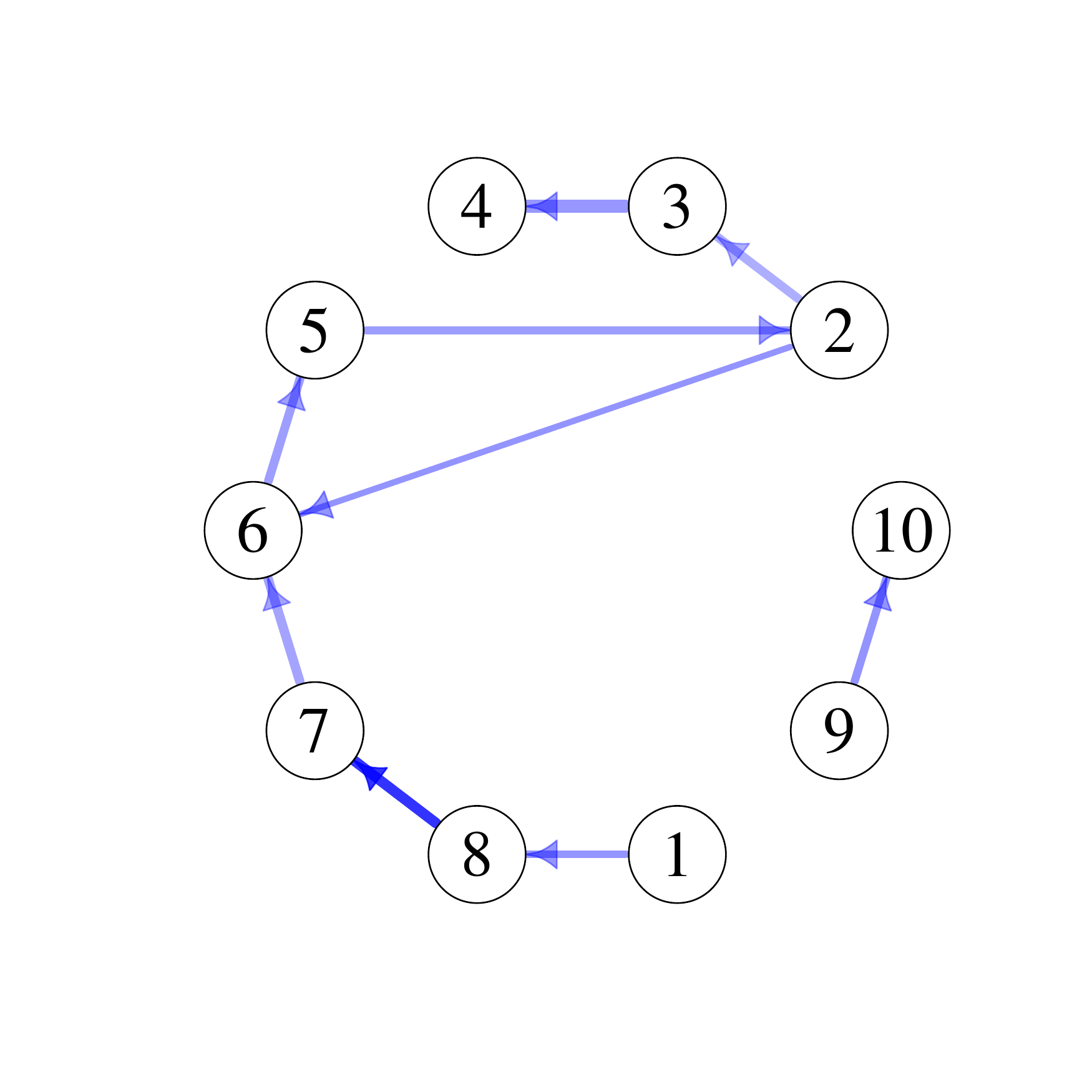}
}} \\
&\hspace{-2mm} \tiny $\bf \mathbf{SHD} = 2$ &\hspace{-2mm} \tiny$\bf \mathbf{SHD} = 1$&\hspace{-2mm} \tiny$\bf \mathbf{SHD} = 2$ &\hspace{-2mm} \tiny$\bf \mathbf{SHD} = 12$ &\hspace{-2mm} \tiny$\bf \mathbf{SHD} = 3$
 \vspace{-0.15cm} \\
  \end{tabular}
 \captionof{figure}{{\small Synthetic data. Stability selection results for \ice with parameters $\mathbb{E}(V) = 2$ and $\pi_{thr} = 0.75$. The intensity of the edges illustrates the relative magnitude of the estimated coefficients, the width shows how often an edge was selected. The edge from node 6 to node 10 is associated with the smallest coefficient in absolute value. It is retained in none of the settings in the stability selection procedure.}} \label{fig:sim_results_stabSel}
\end{table}

\end{document}